\documentclass[11pt]{article}
\usepackage{geometry}
\usepackage{amsmath, amssymb, amsthm}
\usepackage{epsfig}
\usepackage{graphics}
\usepackage{graphicx}
\usepackage{hyperref}
\usepackage{accents}
\usepackage{epstopdf}
\usepackage{multirow}
\usepackage{subfigure}
\usepackage{color}
\usepackage{subfigure}
\usepackage[utf8]{inputenc}
\usepackage{latexsym}
\newtheorem{theorem}{Theorem}[section]
\newtheorem{remark}[theorem]{Remark}
\usepackage{cite}
\usepackage{authblk}

\usepackage{subfiles}

\usepackage{xr}
\usepackage{cleveref}
\makeatletter
\newcommand*{\addFileDependency}[1]{
\typeout{(#1)}
\@addtofilelist{#1}
\IfFileExists{#1}{}{\typeout{No file #1.}}
}\makeatother

\newcommand*{\myexternaldocument}[1]{%
\externaldocument{#1}%
\addFileDependency{#1.tex}%
\addFileDependency{#1.aux}%
}
\myexternaldocument{supplementary}

\usepackage{textcomp,gensymb}
\usepackage{tabularx} 

\usepackage{authblk} 
\usepackage{orcidlink} 

\begin{document}
   
\title{The silent threat of methane to ecosystems: Insights from mechanistic modelling}


\author{Pranali Roy Chowdhury $^*$ \orcidlink{0000-0002-6048-0718} , Tianxu Wang\orcidlink{0000-0002-1911-4396}, Shohel Ahmed \orcidlink{0000-0002-1639-0950}, Hao Wang  \footnote{Corresponding authors: pranali@ualberta.ca; hao8@ualberta.ca }}

\affil[]{Department of Mathematical and Statistical Sciences, University of Alberta, Edmonton, Canada }


\date{}

\maketitle

\begin{center}
{\bf Abstract}
\end{center}


Over the past century, atmospheric methane levels have nearly doubled, posing a significant threat to ecosystems. Despite this, studies on its direct impact on species interactions are lacking. Although bioaccumulation theory explains the effects of contaminants in trophic levels, it is inadequate for gaseous pollutants such as methane. This study aims to bridge the gap by developing a methane-population-detritus model to investigate ecological impacts in aquatic and terrestrial ecosystems. Our findings show that low methane concentrations can enhance species growth, while moderate accumulation may induce sub-lethal effects over time. Elevated methane levels, however, lead to ecosystem collapse. Furthermore, prolonged exposure to the gas increases the sensitivity of species towards rising temperatures. Multiscale analysis reveals that rapid methane accumulation leads to long transients near the extinction states. We argue that high emission rates can push the system towards a critical threshold, where the ecosystem shifts to an alternative stable state characterized by elevated methane concentrations. This work highlights the urgent need for a better understanding of the fatal role of methane in ecosystems for developing strategies to mitigate its effects amid climate change.

\vspace{1.0cm}
	

\vspace{1cm}
\newpage 


	

	
\newpage

\section{Introduction}

Anthropogenic activities across the globe release substantial amounts of methane into the environment, significantly contributing to environmental pollution \cite{Hader2020,MethaneBudget}. Elevated methane concentrations not only accelerate global warming but also disrupt ecological processes, such as altering species distribution, reducing biodiversity, and shifting ecosystem dynamics \cite{Bellard2012}. 
Methane ($\mathrm{CH}_4$) is emitted from various sources such as oil sand tailing ponds, wetlands, wildfires, and geological seeps. However, human activities, including agriculture, fossil fuel extraction and consumption, landfill decomposition, and biomass combustion, are major contributors to methane emissions \cite{aydin2012methane}. During the production and transport of coal, natural gas, and oil, as well as through the dumping and processing of organic waste and wastewater a significant amount of methane is emitted into atmosphere \cite{karakurt2012sources}. Additionally, the anaerobic decomposition of organic materials in landfills or in water bodies contaminated by industrial waste leads to significant methane release into the atmosphere \cite{siddique2011anaerobic}. Figure.~\ref{Fig:Methane_sources} collectively describes the emission from major sources. 

\begin{figure}[ht!]
    \centering
    \includegraphics[width=0.8\linewidth]{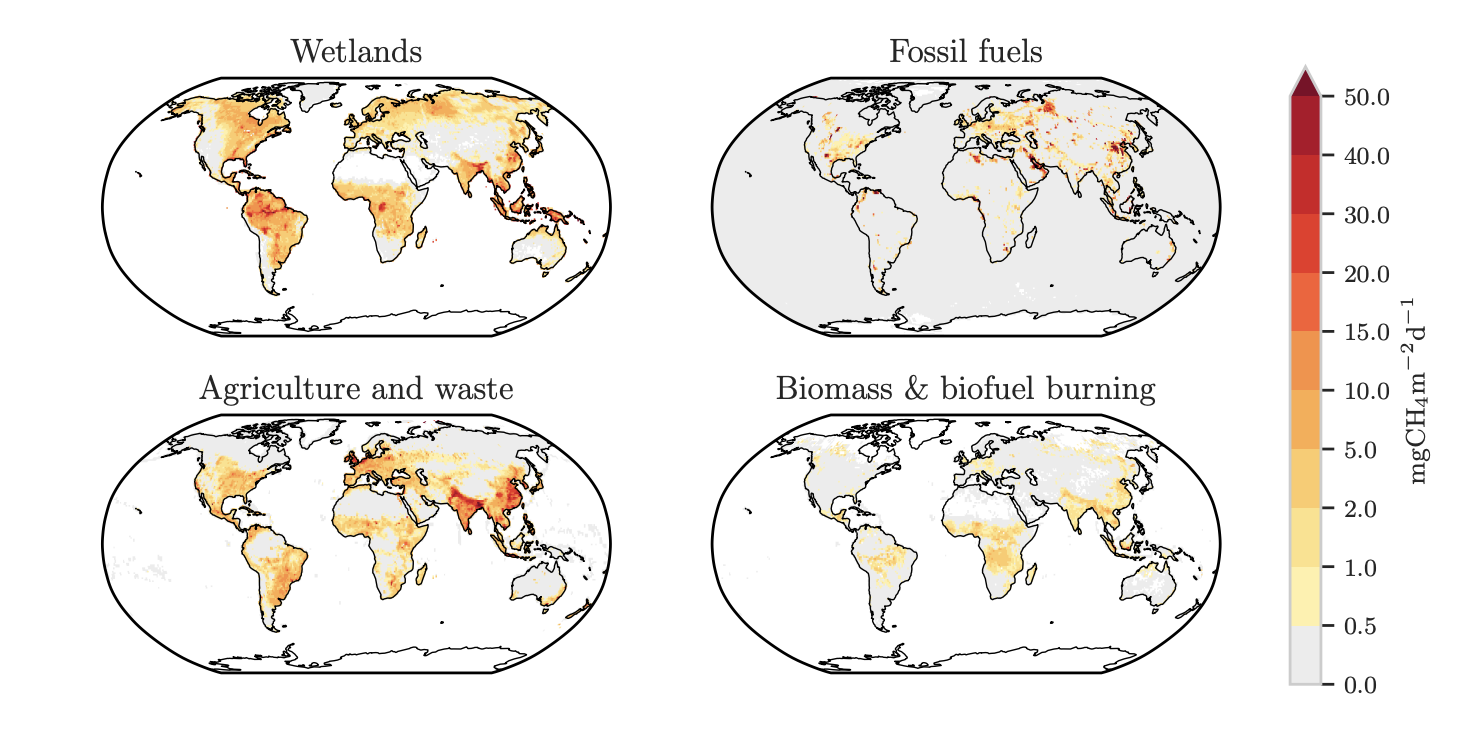}
    \caption{Global methane emission from major source categories from 2010-2019 adapted from \cite{MethaneBudget}.}
    \label{Fig:Methane_sources}
\end{figure} 

Methane emissions not only contribute to climate change but also alter the chemical composition of ecological species, influencing ecological interactions and the population dynamics of key species \cite{Rosenblatt2016, Beever2017, Briffa2012,paerl2018global}. Mitigating methane emissions and nutrient loading is crucial for maintaining ecological balance and preserving the functional roles of species within ecosystems. In recent decades, greenhouse gases have garnered significant attention due to their strong association with rising environmental temperatures. Among the primary, greenhouse gases ($\mathrm{CO}_2,\,\,\mathrm{CH}_4,\,\mathrm{N_2O}$), methane $(\mathrm{CH}_4)$ is the second most abundant GHGs after carbon dioxide $(\mathrm{CO}_2).$ However, the comparative impact of $\mathrm{CH}_4$ is 28 times greater than $\mathrm{CO}_2$ over a 100-year period \cite{fernandez2020methane}. 
While $\mathrm{CO}_2$ emissions have been regulated by governments in many countries, methane emissions have been underestimated over the past decades. This oversight has allowed some industries to exploit regulatory gaps, deliberately emitting methane to avoid government inspections. 

Research on toxicity has traditionally focused on contaminants such as heavy metals \cite{jaishankar2014toxicity}, pesticides, and organic pollutants \cite{Fleeger}. For example, studies have investigated the effects of sulfur dioxide (\(\mathrm{SO}_2\)) on plant growth, photosynthesis, and respiration \cite{garsed1981sulfur}, as well as the impacts of nitrogen dioxide (\(\mathrm{NO}_2\)), \(\mathrm{SO}_2\), and ozone (\(\mathrm{O}_3\)) on radish plants \cite{reinert1981radish}. 
Several studies examined the impacts of zinc (Zn), copper (Cu), and mercury (Hg) on aquatic populations \cite{armstrong1979mercury, widdows1997mussels}. While these efforts have advanced our understanding of contaminant toxicity, dissolved methane remains an overlooked area of study. Despite being recognized as a potent greenhouse gas, its potential toxicological effects within aquatic ecosystems have received little attention. Considering its ability to influence oxygen availability, microbial activity, and overall ecosystem stability, it highlights the urgent need for detailed toxicity assessments. Future research should prioritize the evaluation of methane toxicity to better understand its broader environmental consequences. 

Ecological risks are often assessed through individual-level responses like fecundity and mortality under controlled conditions, which may not fully reflect ecosystem-level impacts. While population-level toxicity tests are costly and impractical, as a result, mathematical modelling has emerged as a valuable tool in ecotoxicology. 
Different types of mathematical models including toxicity-extrapolation models (e.g., \cite{pastorok2003role,pastorok2016ecological}), Toxicokinetic-Toxicodynamic (TKTD) models (e.g., \cite{revel2024tktd}), matrix population models (e.g., \cite{spromberg2006toxicity,hayashi2009population}), bioaccumulation models (e.g., \cite{arnot2004bioaccumulation, mathew2008pah} and the ordinary and partial differential equation models (e.g., \cite{Huang13,huang2015impact,wang2024stoichiometric,Zhou,MISRA20138595,TAKHIROV2024100414,Karim2024}) have been developed to study the impact of environmental toxicants on the dynamics of exposed populations in polluted ecosystems. These works highlight the chemical risks across biological hierarchies, from cells to entire ecosystems. 

Existing toxin-based models typically assume that higher-trophic organisms primarily absorb toxins through food ingestion, accumulating more pollutants as they move up the food chain \cite{Hallam83,Thieme,Freedman91}. This bioaccumulation is commonly observed with contaminants such as microplastics and chemical toxins \cite{huang2015impact,wang2024stoichiometric}. However, methane, as a gas, disperses easily and has not been observed to bioaccumulate. The effects of methane and the sensitivity of organisms at different trophic levels vary significantly. Therefore, the previous modelling frameworks are not suitable for this work. Moreover, methane plays a beneficial role in primary ecosystems by supporting the growth of photoautotrophs through microbial processes such as methanotrophy, which converts methane into available carbon sources. This dual nature of methane, both a pollutant and a resource, requires a specialized modelling approach to accurately capture its ecological effects.

The pathways of methane emission are highly complex to incorporate into a single modelling framework. Therefore, in this paper, we present a simple mechanistic approach to model the interactions between methane, resources, consumers, and detritus in the ecosystem using differential equations. The dynamics of methane concentration are captured with two primary sources: external inputs from human activities and methane production through the decomposition of organic matter. We discuss changes in the density of consumers and resources, with their growth and development influenced by both methane concentration and temperature, and study the decomposition of organic matter through changes in detritus density. The primary aim of this work is to provide a framework for understanding how rising methane levels impact trophic dynamics. This work serves as an initial step towards investigating the ecological impact of methane.



\section{Model and Methods}\label{sec:Methods}
\subsection{Model formulation}\label{sec:Model formulation}
In this section, we develop a conceptual model to investigate the influence of methane on a resource-consumer system with the help of differential equations. The model consists of four state variables: the concentration of methane gas $(M),$  density of resources $(X),$ density of consumer $(Y)$, and the density of detritus $(D)$ at a particular time $t.$ A schematic sketch of our model is shown in Fig.~\ref{fig:framework}. The rate of change in environmental methane concentration is based on the mass balance equation. We consider the rate of production of methane from anthropogenic activities as $E_{\mathrm{in}}(M)$, while the methane generated from the decomposition of dead organic matter is denoted as $E_D(D)$. The natural loss of methane occurs through various pathways including diffusion, uptake by organisms, and oxidation by methane-oxidizing bacteria (MOB). In oxygenated conditions, MOB converts methane to carbon dioxide and water, removing 45–100\% of methane in lake ecosystems \cite{Bastviken23}. In aquatic ecosystems, methane majorly escapes via ebullition \cite{Bastviken04}, while in terrestrial systems, it diffuses into air or water.


\begin{figure}[ht!]
    \centering
    \includegraphics[width=0.5\linewidth]{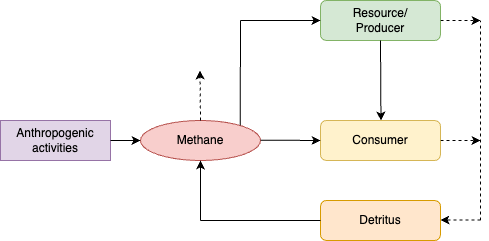}
    \caption{Schematic diagram of the model.}
    \label{fig:framework}
\end{figure}

For a more general applicability of our model, we discuss the assumptions used to formulate the functional forms of the model in the case of both the aquatic and terrestrial ecosystems. The dumping of industrial effluent, agricultural runoff from fields, or leaks from the oil and gas industries often enter nearby water bodies or diffuse to higher air levels. In aquatic ecosystems, water has limited capacity to dissolve gases. Let \( \Bar{M} \) denote the maximum methane concentration that water can dissolve \cite{Grabowska22}, and assume a constant external methane input rate \( M_{\mathrm{in}} \). The external methane input function for aquatic ecosystems is then expressed as 
\[
E_{\mathrm{in}}(M) = \max\{c_1(\Bar{M} - M), 0\}M_{\mathrm{in}},
 \]
where \( c_1 \) represents the dissolution rate of methane in water.  
In contrast, terrestrial ecosystems have a much greater capacity to absorb gases. As a result, the input rate is approximately constant, and the input function simplifies to:  
\[
E_{\mathrm{in}}(M) = M_{\mathrm{in}}.
\]
Assuming the bacterial decomposition of the detritus at a rate $d_4(D) = pD,$ the emission of methane from the dead organic matter is determined by $E_D(D) = e_1pD$ \cite{Bartosiewicz}. 

The emission of methane from ruminant animals depends linearly on their body size \cite{Franz11}, and small herbivores produce a negligible amount of methane. Therefore, we consider $E_P(X,Y)$ as a negligible quantity for our purpose. Further, considering the small size of the consumers in the lower trophic level, the uptake of methane is also insignificant\cite{Hallam84}. For aquatic ecosystems, we consider two possible pathways of loss of methane: oxidation of $\mathrm{CH}_4$, and loss from water surface to atmosphere. Studies reveal a linear relationship between methane dissolved in water and the oxidation rate. Therefore, we have $$d_1(M) = -M_{\mathrm{out}}M,$$ where $M_{\mathrm{out}}$ is the net rate of loss of methane from the lake to the atmosphere from various pathways which includes ebullition flux and oxidation \cite{Bastviken04}. For terrestrial ecosystems, we also assume a linear output $d_1(M) = -M_{\mathrm{out}}M$. The growth function of prey or primary producers, \( g_1 \), is modelled using the classical logistic growth function, where the maximal growth rate depends on both temperatures (\( T \)) and methane concentration (\( M \)). The temperature-dependent photosynthetic activity of aquatic and terrestrial autotrophs is well-studied in the literature and follows an unimodal functional form:  
$$
\dfrac{r}{1 + \gamma_1(T - T_X)^2},
$$ 
where the maximum growth occurs at an optimal temperature (\( T_X \)), and \( r \) and \( \gamma_1 \) represent the maximal growth rate and the influence of temperature, respectively. Under oxic conditions, methane oxidation occurs through MOB contributing to accumulation of inorganic carbon. The methane oxidation rate exhibits a linear relationship with the dissolved methane concentration \cite{DAngelo}. The carbon dioxide produced during methane oxidation is subsequently taken up by phytoplankton in the presence of light during photosynthesis, thereby indirectly supporting the growth of phytoplankton \cite{Cerbin}. The mechanism is illustrated in Fig.~\ref{fig:MOB_mechanism}.  
\begin{figure}[ht!]
    \centering
    \includegraphics[width=10cm, height=3cm]{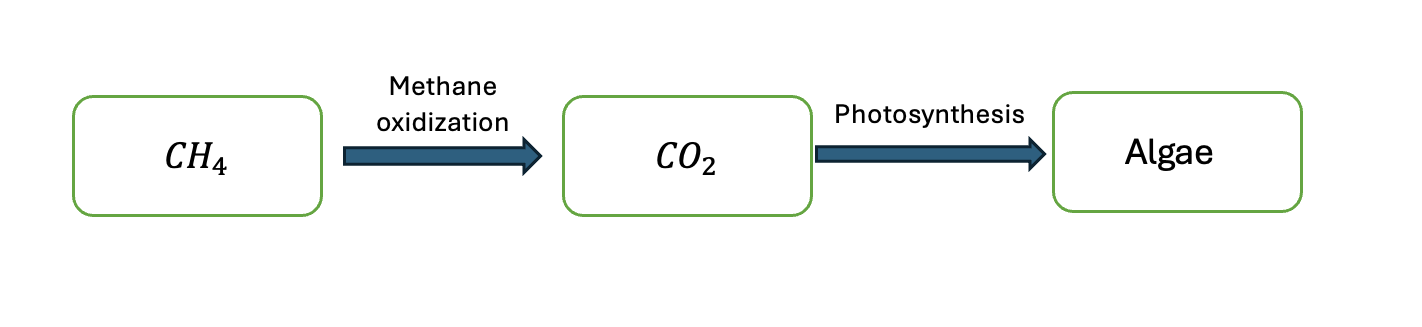}
    \caption{Mechanism representing the growth of algae due to methane oxidation.}
    \label{fig:MOB_mechanism}
\end{figure}
This suggests that the growth of primary producers increases with methane oxidation \cite{Enebo,Hadiyanto}. We assume this dependence to be linear, modeled as \( 1 + \gamma_M M \), where \( \gamma_M \) represents the influence of methane on growth. The consumer-resource interaction, \( f(X) \), is modeled using the classical Holling type-II functional response \cite{Holling} as
\[
f(X) = \dfrac{\alpha X}{\beta + X},
\]  
where \( \alpha \) and \( \beta \) are parameters describing the consumer attack rate and the half-saturation constant, respectively. 

Assuming that the conversion efficiency of ectotherm species \cite{Amarasekare12} is temperature-dependent, we consider  
\[
e(T) = \dfrac{e}{1 + \gamma_2(T - T_Y)^2},
\]  
where the maximum growth of predators, driven by grazing on prey or primary producers, occurs at an optimal temperature \( T_Y \).  
In this work, we assume \( e \) to be independent of methane. This assumption is biologically reasonable based on experimental studies on the growth rate and reproduction of Daphnia enriched with biogenic methane \cite{Kankaala06}.  
The predators' mortality rate due to the influence of methane can be modeled using a power law of the form \cite{Huang13}
\[
d_Y + d_M M^l,
\]  
where \( d_Y \) represents the natural mortality of predators, \( d_M \) quantifies mortality due to methane, and \( M \) is the methane concentration. For simplicity, we consider the special case of \( l = 1 \), leading to the mortality function 
\[
d_3(M, Y) = d_Y + d_M M,
\]  
where \( d_Y \) is the natural mortality rate of predators, and \( d_M \) quantifies the impact of methane on mortality.  
Since dead organic matter contributes to the detritus pool \cite{QAn, Andrew}, we model the input to detritus as  
\[
g_2(X, Y, M) = d_X X + \big(d_Y + d_M M\big)Y,
\]  
where \( d_X \) is the mortality of prey or primary producers, and \( Y \) represents the consumer population. Based on these assumptions, we propose a novel model applicable to both aquatic and terrestrial ecosystems.  

\begin{equation}\label{Eq:Model}
    \begin{aligned}
          \frac{\mathrm{d} M}{\mathrm{d} t} &=f(M,D),\\
        \frac{\mathrm{d} X}{\mathrm{d} t} &= \frac{r}{1+ \gamma_1 (T-T_X)^2} \left( 1 + \gamma_M M \right) X\left(1-\frac{X}{K}\right) - \frac{\alpha X}{\beta + X}Y  - d_X X,\\
        \frac{\mathrm{d} Y}{\mathrm{d} t} &= \frac{e}{1+ \gamma_2 (T-T_Y)^2}\frac{\alpha X}{\beta + X}Y  - (d_Y + d_M M) Y,\\
         \frac{\mathrm{d} D}{\mathrm{d} t} & =  d_X X + (d_Y + d_M M) Y - pD.
    \end{aligned}
\end{equation}
For the aquatic ecosystem, we define \( f(M,D) \) as  
\begin{equation}\label{Eq:aquatic_eq}
f(M,D) = \max\{c_1\left(\Bar{M} - M\right), 0\}(M_{\text{in}} + e_1 p D) - M_{\text{out}} M,
\end{equation}
where \( \Bar{M} \) is the maximum methane concentration that water can dissolve. For the terrestrial ecosystem, we take \( f(M,D) \) to be
\begin{equation}\label{Eq:terrestrial_eq}
    f(M,D) = M_{\text{in}} + e_1 p D - M_{\text{out}} M.
\end{equation}

\subsection{Model parameterization}

The model in \eqref{Eq:Model} represents a general methane-prey-predator-detritus system. 
For the model in \eqref{Eq:Model} with \eqref{Eq:aquatic_eq}, we parameterize the system using algae and daphnia as representative prey and predators in the aquatic ecosystem. However, studies investigating methane's impact on the ecosystem are scarce. Consequently, the parameters are chosen based on the general prey-predator systems in the literature, while some are estimated based on available data, as detailed in SI \ref{sec:Parameterization}.

\subsection{Multiple timescales analysis} \label{sec:Multiple timescales framework of the model}

The continuous input of methane into the water is assumed to be faster in highly active industrial areas than the biological processes like the growth and development of species. In contrast, detritus generation and decomposition occur at a much slower time. This creates distinct differences in timescales in the system. In this section, we analyze the model by decomposing it into subsystems, each capturing dynamics at a specific timescale \cite{Kuehn15}. While we use the aquatic ecosystem model as an example to explore these multiscale dynamics, a similar approach can be applied to terrestrial models. We non-dimensionalize the model \eqref{Eq:Model} to focus on the key parameters that govern the system's behaviour. The details of the non-dimensionalization can be found in the  SI \ref{sec: Nondimensionalization}.

\begin{equation}\label{Eq:three_timescale_model}
\begin{aligned}
     \frac{\mathrm{d} m}{\mathrm{d} t} &= \zeta\max\{ \left(\theta-m\right),0\}(1 + e_1w) - \sigma m \equiv
      F_1(m,w),\\
     \frac{\mathrm{d} u}{\mathrm{d} t} &= \varepsilon_1\left(\rho_1 u\left( 1 + m \right)\left(1-u\right) - \frac{ \tau uv}{\kappa + u}  - \mu_1 u\right) \equiv
      \varepsilon_1 F_2(m,u,v),\\
     \frac{\mathrm{d} v}{\mathrm{d} t} &= \varepsilon_1\left(\rho_2  \frac{\tau uv}{\kappa + u}  - \left(\mu_2 + \eta m\right) v\right) \equiv
      \varepsilon_1 F_3(m,u,v),\\
    \frac{\mathrm{d}w}{\mathrm{d} t} & =  \varepsilon_2\left( \mu_1 u + \left(\mu_2 + \eta m\right) v - \delta w\right)\equiv
      \varepsilon_2 F_4(m,u,v,w).         
\end{aligned}
\end{equation}

\subsubsection{Dynamics of subsystems}
To investigate the dynamics of the non-dimensionalized model \eqref{Eq:three_timescale_model}, we decompose it into subsystems and analyze it with the help of geometric singular perturbation theory \cite{Fenichel,Krupa01b}.
\paragraph{Fast dynamics} The layer subsystem or the fast system is obtained for $\varepsilon_1,\varepsilon_2 \rightarrow 0$ in the system \eqref{Eq:three_timescale_model}, which is given by 
\begin{equation}\label{eq:fast_dynamics}
    \begin{aligned}
        &\frac{\mathrm{d} m}{\mathrm{d} t} = \zeta\max\{ \left(\theta-m\right),0\}(1 + e_1w) - \sigma m \equiv
      F_1(m,w),\\
     &\frac{\mathrm{d} u}{\mathrm{d} t} =0,\,
     \frac{\mathrm{d} v}{\mathrm{d} t} = 0,\,\
    \frac{\mathrm{d}w}{\mathrm{d} t}  =  0.
    \end{aligned}
\end{equation}
This implies that for $u,\,v$ and $w$ fixed at a constant value, we obtain a fast transition in the concentration of methane dissolved in water.
The set of all equilibrium points of the fast subsystem is given by the set $C^1 = \{(m,u,v,w): F_1(m,w) = 0\}, $ which we denote as critical manifold. We can write $F_1(m,w)=0$ in the form of $$w:=\phi(m) = \frac{1}{e_1}\frac{m(\sigma+\zeta)-\zeta\theta}{\zeta(\theta-m)},$$ considering $m<\theta.$ However, to visualize this critical manifold in the higher-dimensional space is difficult. Thus, we study the system \eqref{eq:fast_dynamics} for some fixed values of $w=w_0,$ and the fast flow is determined by the differential equation of $m$. The critical manifold is hyperbolic as $\phi'(m) = \frac{\theta \sigma}{e_1\zeta(m-\theta)^2}>0$ for all the biologically feasible parameter values, and for a fixed value of $w=w_0,$ we have 
$$\frac{dF_1}{dm} = \left\{\begin{array}{ll}-\zeta(1+e_1w_0)-\sigma, &0< m<\theta\\
-\sigma, & m>\theta \\
\end{array}\right. $$
In both cases, the critical manifold \(C^1\) is an attracting manifold. This means that trajectories starting from initial conditions \((m_0, u_0, v_0)\) are rapidly attracted to \(C^1\), as shown in Figure~\ref{fig:fast_flow}(a). Consequently, methane concentration (\(m\)) increases abruptly from its initial value \(m_0\) within a short time, as illustrated in Figure~\ref{fig:fast_flow}(b). Biologically, this implies that when the amount of detritus is fixed, methane concentration increases rapidly until it reaches the solubility limit of water.

\begin{figure}[ht!]
    \centering
    \subfigure[]{\includegraphics[width=0.35\linewidth]{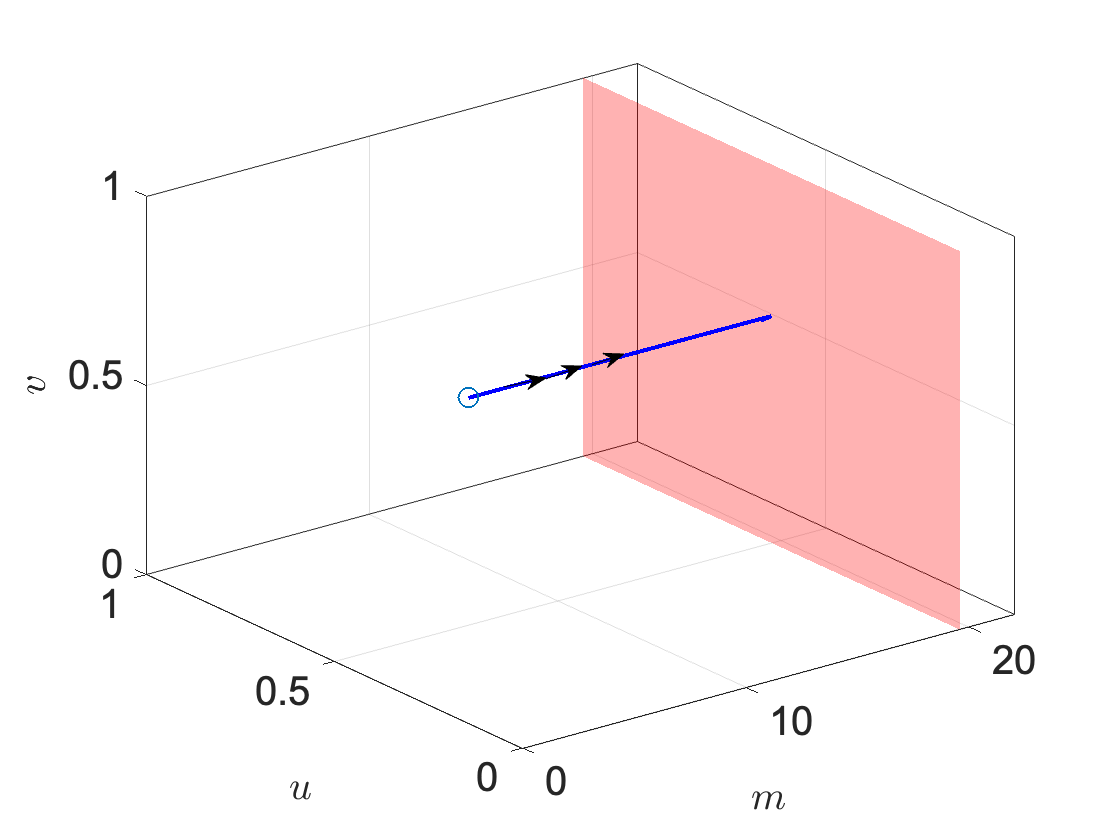}}
    \subfigure[]{\includegraphics[width=0.35\linewidth]{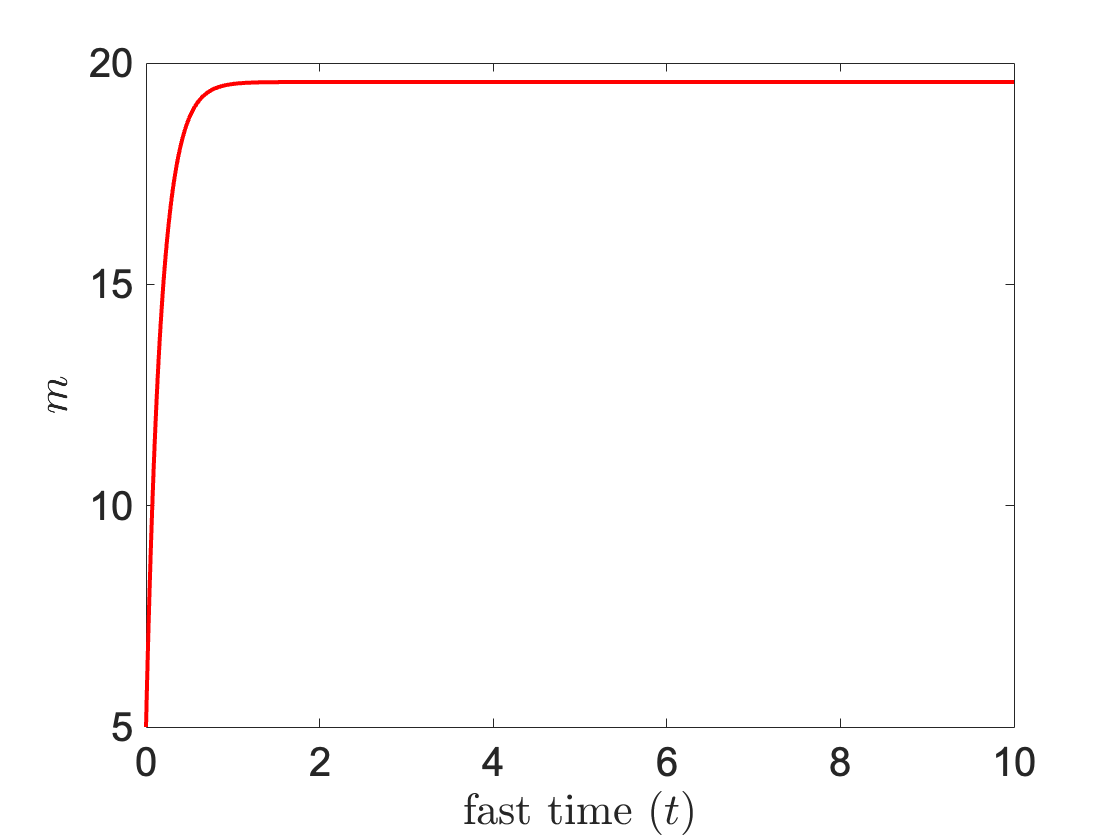}}
\caption{(a) The critical manifold \(C^1\) in the \(muv\) space is shown in red for a fixed \(w_0\), with the trajectory converging toward it shown in blue. Triple arrows indicate the fast flow toward \(C^1\). (b) The corresponding rapid rise in methane concentration (\(m\)) in water is shown with respect to the fast timescale \(t\).}
    \label{fig:fast_flow}
\end{figure}

\paragraph{Intermediate dynamics} For $w=w_0$, the intermediate dynamics is obtained from  intermediate subsystem (SI. Eqn.~\eqref{Eq:intermediate_system}):

\begin{equation}\label{Eq:intermediate_model}
\begin{aligned}
   0 &= \zeta\max\{ \left(\theta-m\right),0\}(1 + e_1w_0) - \sigma m ,\\
     \frac{\mathrm{d} u}{\mathrm{d} t_1} &= \left(\rho_1 u\left( 1 + m \right)\left(1-u\right) - \frac{ \tau uv}{\kappa + u}  - \mu_1 u\right),\\
     \frac{\mathrm{d} v}{\mathrm{d} t_1} &= \left(\rho_2  \frac{\tau uv}{\kappa + u}  - \left(\mu_2 + \eta m\right) v\right).  
\end{aligned}
\end{equation} 
The solutions \((m, u, v)\) describe the flow constrained on the \(C^1\) manifold for \(w = w_0\). Here, $t_1$ is the intermediate timescale, and we have a system with 1-fast and 2-intermediate variables. The intermediate subsystem has three different kinds of equilibrium: The trivial equilibrium corresponding to the complete extinction $(m^*,0,0)=\left(\frac{\zeta\theta(1+e_1w_0)}{\zeta(1+e_1w_0)+\sigma},0,0\right),$ the zooplankton-free equilibrium $$(m^*,u^*,0)= \left(\frac{\zeta\theta(1+e_1w_0)}{\zeta(1+e_1w_0)+\sigma},\frac{\rho_1((1+m^*)-\mu_1)}{\rho_1(1+m_*)},0\right),$$ and a unique coexistence equilibrium $$(m^*,u^*,v^*) = \left(\frac{\zeta\theta(1+e_1w_0)}{\zeta(1+e_1w_0)+\sigma},\frac{\kappa(m^*\eta+\mu_2)}{\rho_2\tau-\mu_2-m^*\eta},\frac{(\kappa+u^*)(\rho_1(1+m^*)(1-u^*)-\mu_1)}{\tau}\right).$$ 
A trajectory starting from a generic point in the state space is first attracted to the critical manifold with a fast flow. Subsequently, it evolves at an intermediate speed along the manifold, guided by the dynamics of the intermediate subsystem. Biologically, this implies that when the detritus amount is fixed, methane concentration increases rapidly, and population dynamics adjust accordingly. However, these adjustments depend on the specific interactions among organisms. Typically, the population dynamics are slower than the methane dynamics.

\paragraph{Slow dynamics} The slow dynamics are given by the reduced subsystem as follows
\begin{equation}\label{eq:slow_dynamics}
    \begin{aligned}
      &F_1(m,w) = 0,\,F_2(m,u,v) =0,\,F_3(m,u,v)=0,\\
    & \frac{\mathrm{d}w}{\mathrm{d} t_2}  =  \left( \mu_1 u + \left(\mu_2 + \eta m\right) v - \delta w\right)\equiv F_4(m,u,v,w).
    \end{aligned}
\end{equation}
The slow flow occurs along the curve of intersections of the above surfaces in four-dimensional space. Since visualizing this curve is difficult, we study the system for a fixed value of \(w = w_0\). Biologically, this simplification is justified because \(w\) evolves on a much slower timescale compared to the other variables.

\section{Results}\label{sec:Numerical Results}
 In this section, we investigate the model \eqref{Eq:Model}-\eqref{Eq:terrestrial_eq} focussing on the two ecosystems: aquatic and terrestrial. For the long-term behaviour of the system, we examine the system behaviour near the steady states. The system can depict three possible outcomes: complete extinction of resource, consumer and detritus; extinction of consumer and coexistence of both species. The analytical description of the steady states and their stability is discussed in SI.~\ref{App:Mathematical_Analysis}. We investigate the impact of methane on the model \eqref{Eq:Model} with \eqref{Eq:aquatic_eq} parameter values given in Table\eqref{tab:Model_para_table_aqua}.


\begin{table}[ht!]\scriptsize
\centering
    \begin{tabular}{ m{2cm}  m{8.5cm} m{3cm}} 
    \hline
 Parameter & Description & Value \\ 
 \hline 
  $T$ & Temperature of the environment & 20  \\ 
  
  $T_X$ & Optimal temperature for the growth of producer & 25 \\  

  $T_Y$ & Optimal temperature for the growth of consumer & 20 \\
  
  $c_1$ &  Rate of dissolution of methane in water & 0.05 \\
  
  $\Bar{M}$ & Saturation of methane in water & 22.7 \\
  
  $M_{\mathrm{in}}$ & The exogenous input of methane & $0-100$ \\
  
  $M_{\mathrm{out}}$ & The natural rate of loss of methane & 1\\
 
  $e_1$ & Proportionality constant  &  0.002 \\

  $r$ & Maximum growth rate of producer & 1.2\\

  $\gamma_M$ & Influence of methane to the growth of producer & 1\\

  $K$ & Maximum carrying capacity & 1\\
  
  $\alpha$ &Maximum predation rate & 0.8\\
  
  $\beta$ & Half-saturation constant& 0.25 \\
  
  $d_X$ & Natural death rate of algae & 0.02\\
  
  $d_Y$ & Natural death rate of daphnia& 0.01 \\

  $e$ & Maximal conversion efficiency of daphnia & 0.55\\
  
  $d_M$ & Influence of methane on the death of daphnia & 0.02 \\
  
  $p$ & Decomposition rate of detritus & 0.05\\

  $\gamma_1$ & Temperature related constant for growth of algae & 1 \\

  $\gamma_2$ & Temperature related constant for growth of daphnia & 1\\
  
 \hline
\end{tabular}
 \caption{Parameters of the model \eqref{Eq:Model}-\eqref{Eq:terrestrial_eq}.}
    \label{tab:Model_para_table_aqua}
\end{table}

\subsection{Methane Toxicity and Temperature Sensitivity}

To illustrate the combined effects of temperature and methane on the persistence and extinction of plankton species, we consider $T$ as the bifurcation parameter and observe the change in the stability of the steady states for different concentrations of $M_{\mathrm{in}}.$ The one-parameter bifurcation diagram implying the change in the behaviour of the system dynamics is described in Fig.~\ref{fig:Temperature_bifurcation} of SI \ref{app:extra_figure}.
Fig.~\ref{fig:Two_parameter_T_Min} shows the change in the densities of algae ~\ref{fig:Two_parameter_T_Min_A} and daphnia ~\ref{fig:Two_parameter_T_Min_D} over different temperature and input rates of methane. Algae densities increase and reach their carrying capacity with rising temperatures and methane levels. In contrast, daphnia thrives near its optimal temperature when methane concentrations are low. For \(M_{\mathrm{in}} < 20\), daphnia density increases due to the abundance of algae. Beyond this threshold, daphnia density declines despite the ample food supply, eventually leading to extinction at higher \(M_{\mathrm{in}}\). The temperature sensitivity due to the toxicity effect holds true in the case of terrestrial ectotherm species. However, for the sake of brevity, we omit showing this here. Our findings signify that higher methane concentrations narrow the temperature range for coexistence and fasten the daphnia extinction. This indicates that prolonged exposure to moderate methane levels increases daphnia sensitivity to rising temperatures.

\begin{figure}[ht!]
    \centering
    \subfigure[Algae]{
    \includegraphics[width=0.45\textwidth]{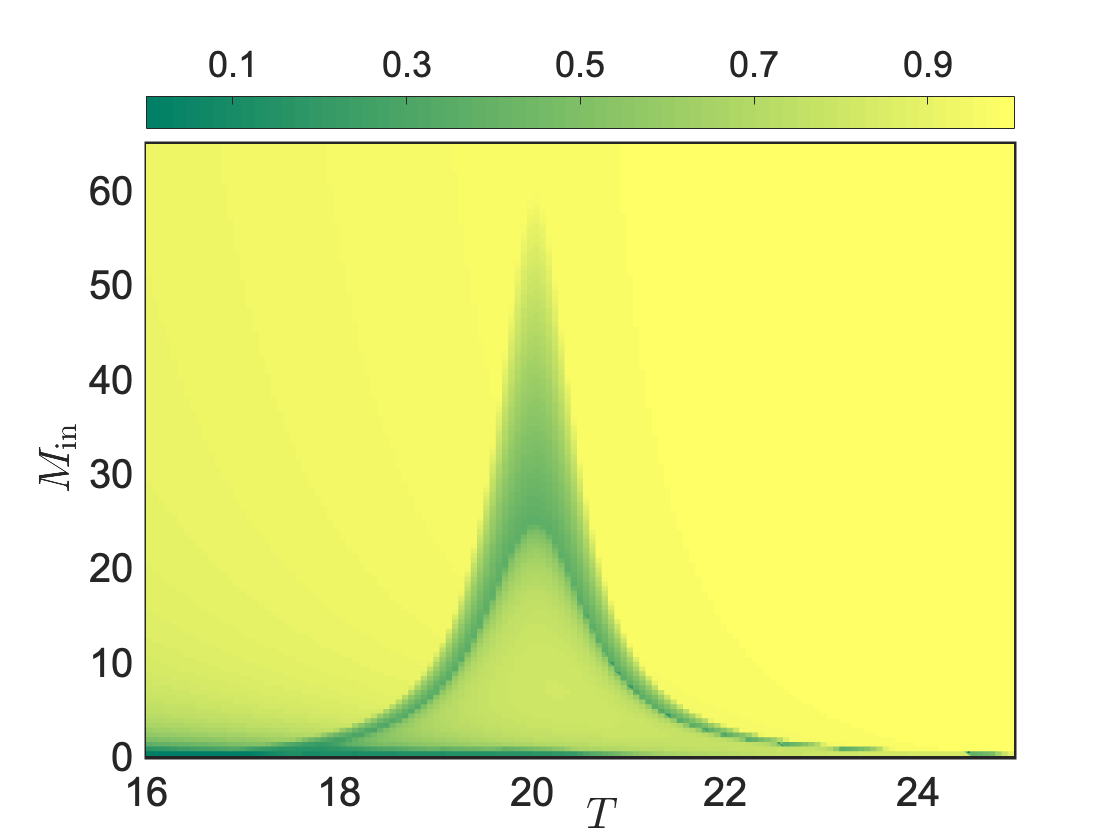}
    \label{fig:Two_parameter_T_Min_A}}
     \subfigure[Daphnia]{
    \includegraphics[width=0.45\textwidth]{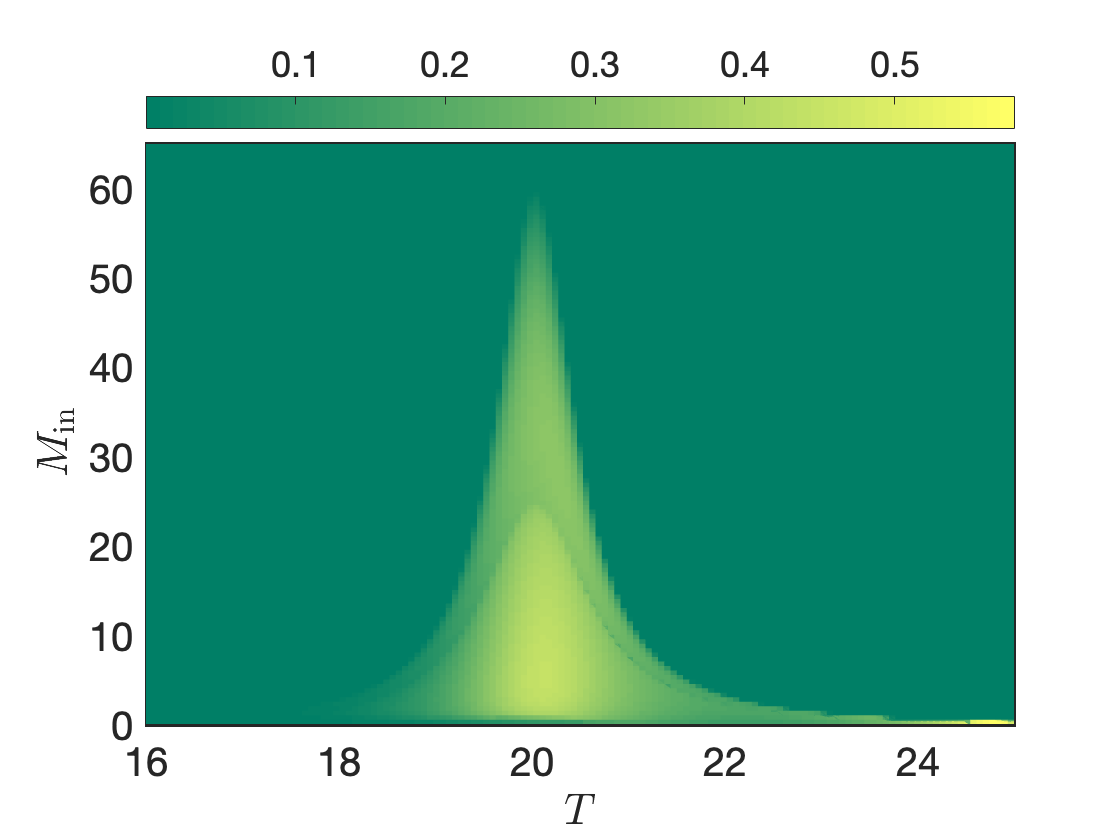}
    \label{fig:Two_parameter_T_Min_D}}
    \caption{The variation in the density of algae and daphnia for different values of temperature and external input rate of methane. }
    \label{fig:Two_parameter_T_Min}
\end{figure}

\subsection{The dual effect of methane on plankton dynamics}
Previously, we observed increased daphnia density at low to moderate \(M_{\mathrm{in}}\) levels at optimal temperature. To further explore the interactions between algae and daphnia at this temperature, we fix the water temperature at \(T = 20\).
\begin{figure}[ht!]
    \centering
    \subfigure[$M_{\mathrm{in}}=5$]{\includegraphics[width=0.3\textwidth]{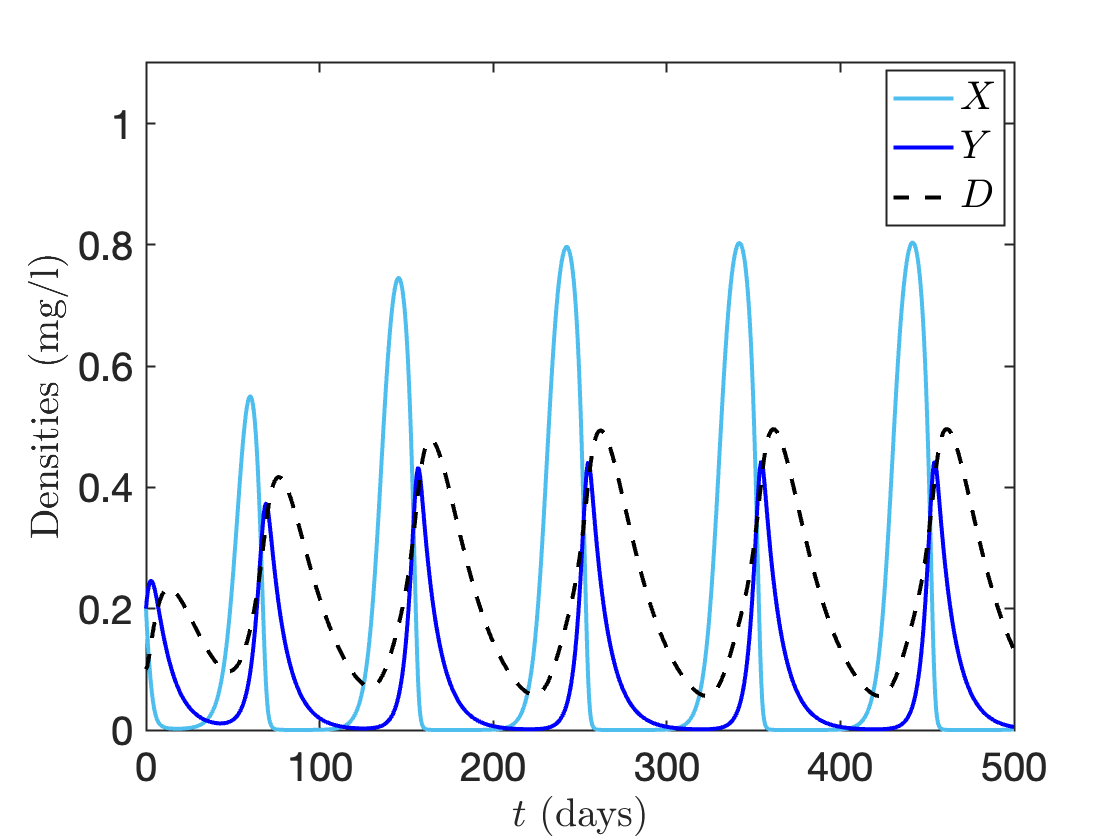}
    \includegraphics[width=0.3\textwidth]{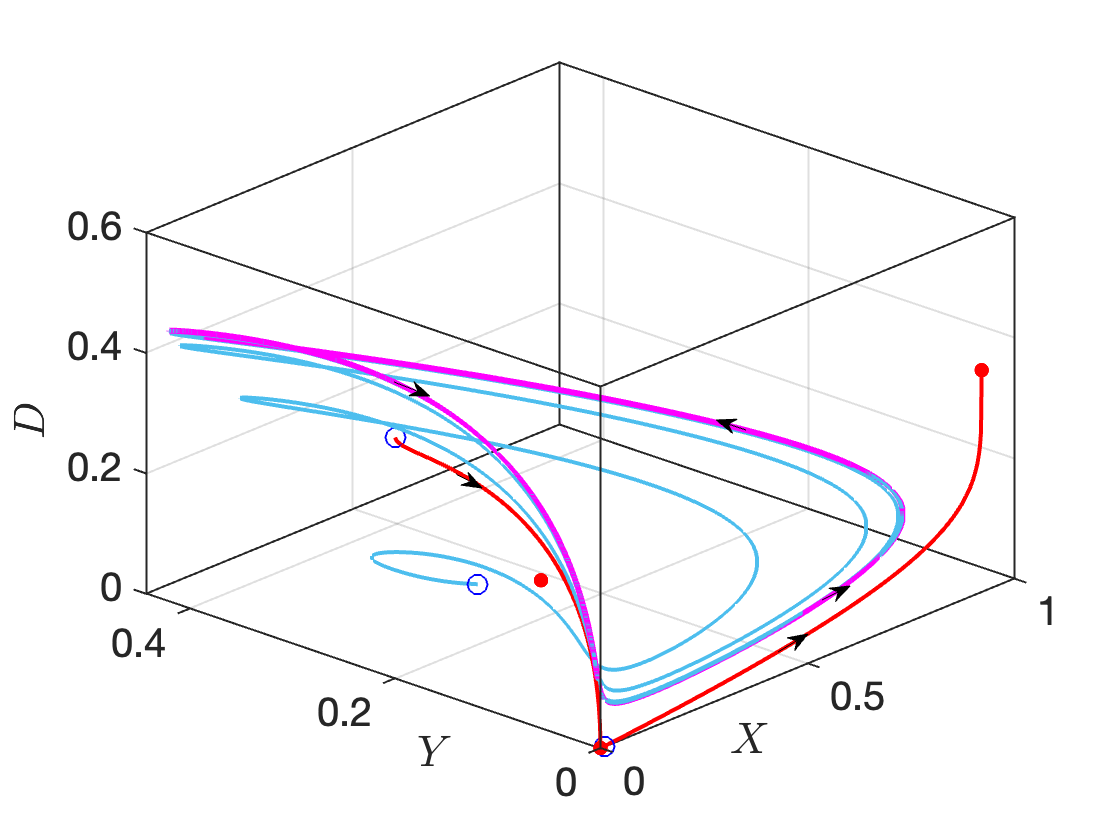}
    }
     \subfigure[$M_{\mathrm{in}}=15$ ]{\includegraphics[width=0.3\textwidth]{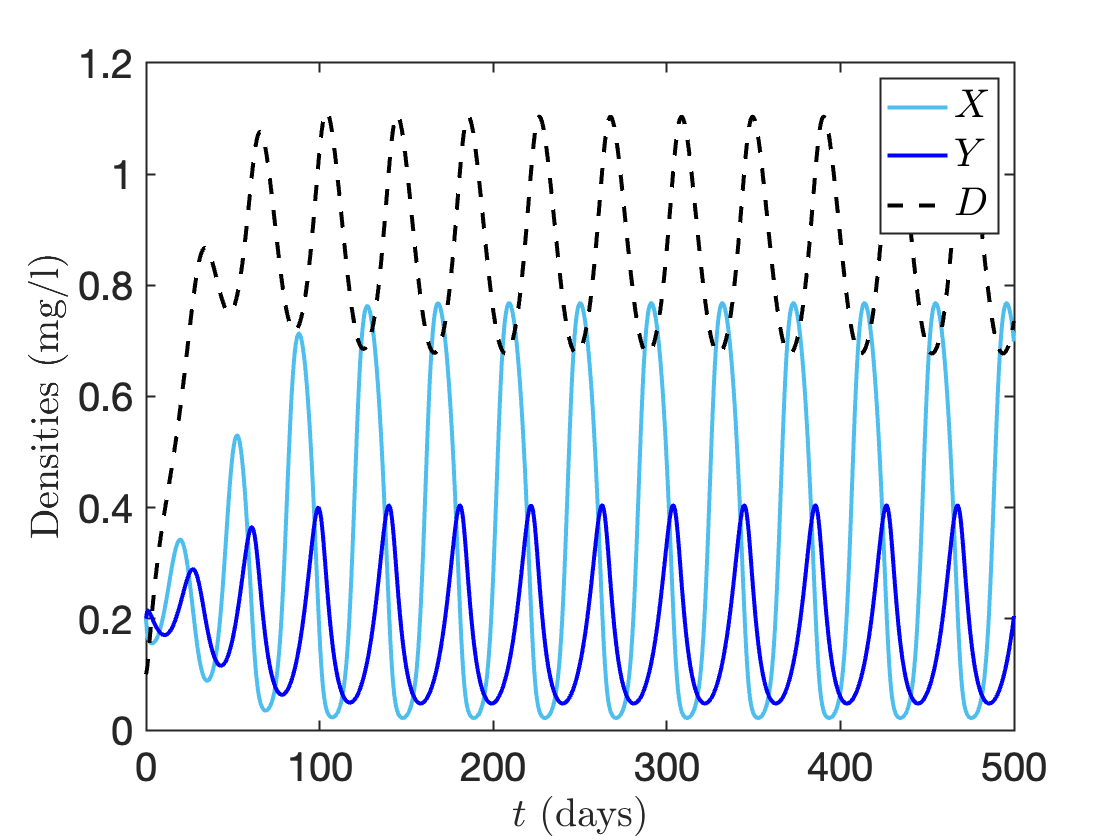}
     \includegraphics[width=0.3\textwidth]{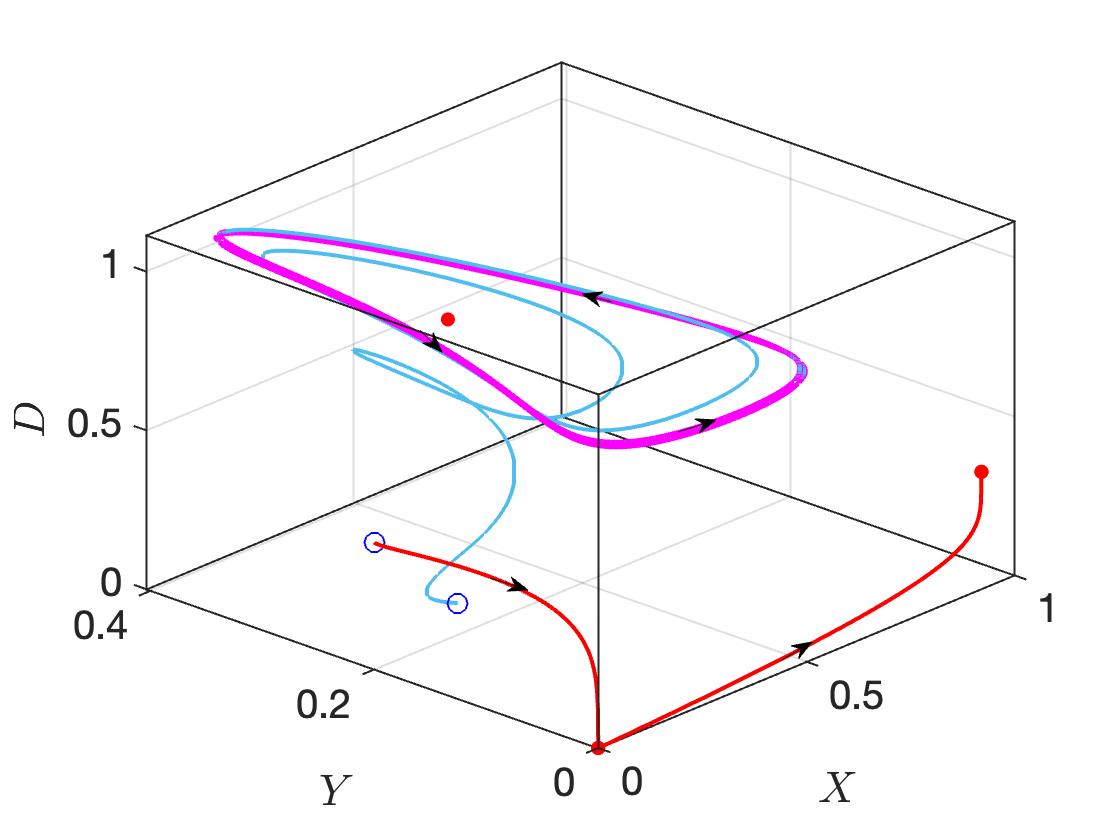}
     }
     \subfigure[$M_{\mathrm{in}}=30$]{\includegraphics[width=0.3\textwidth]{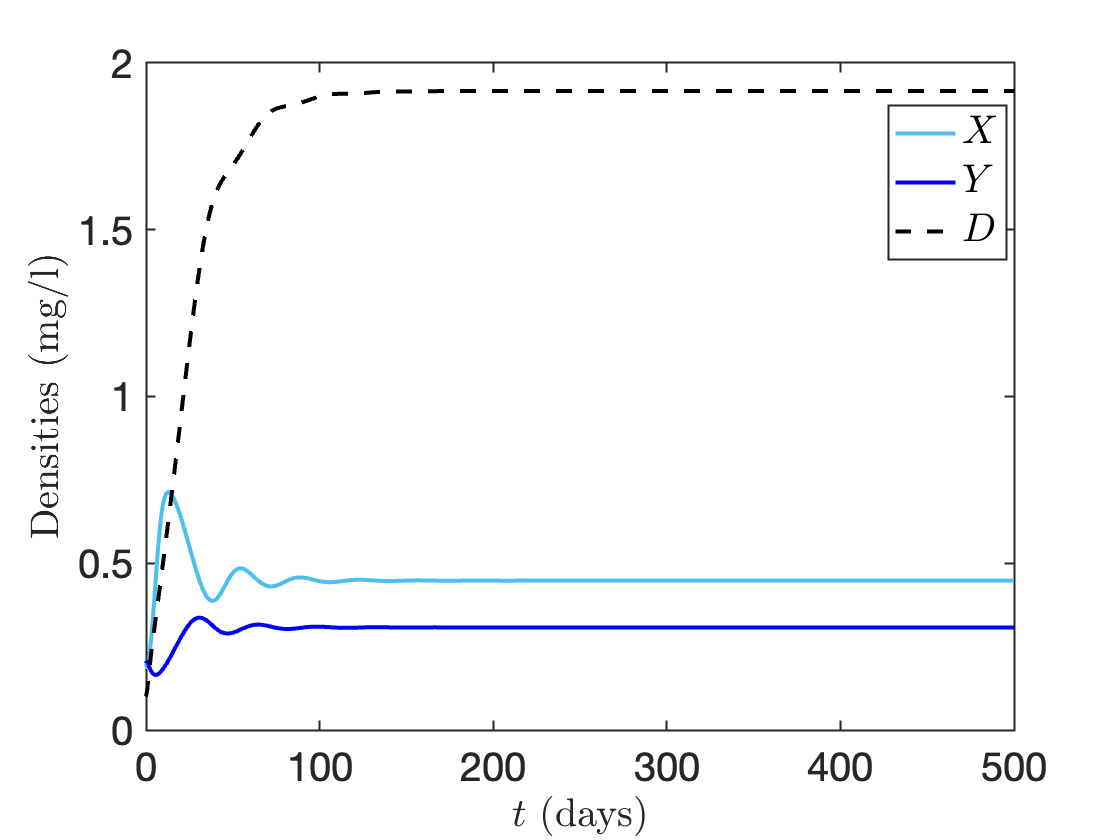}
     \includegraphics[width=0.3\textwidth]{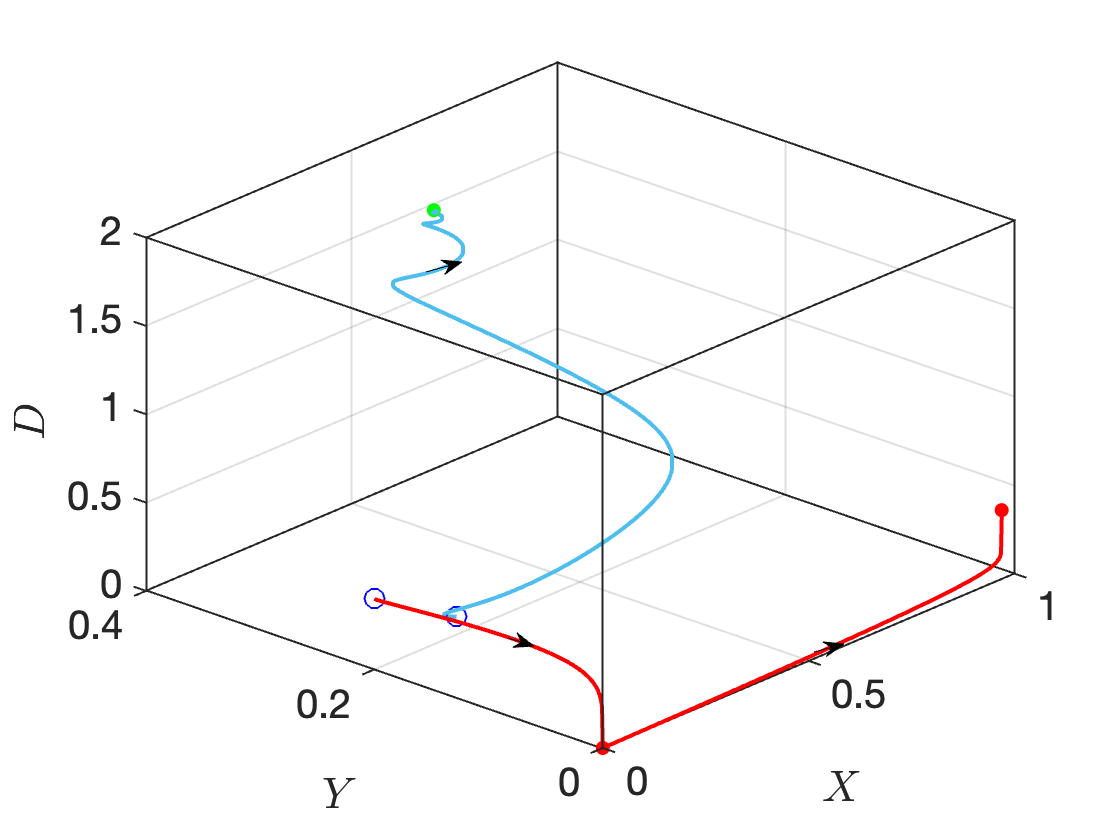}
     }
      \subfigure[$M_{\mathrm{in}}=65$ ]{\includegraphics[width=0.3\textwidth]{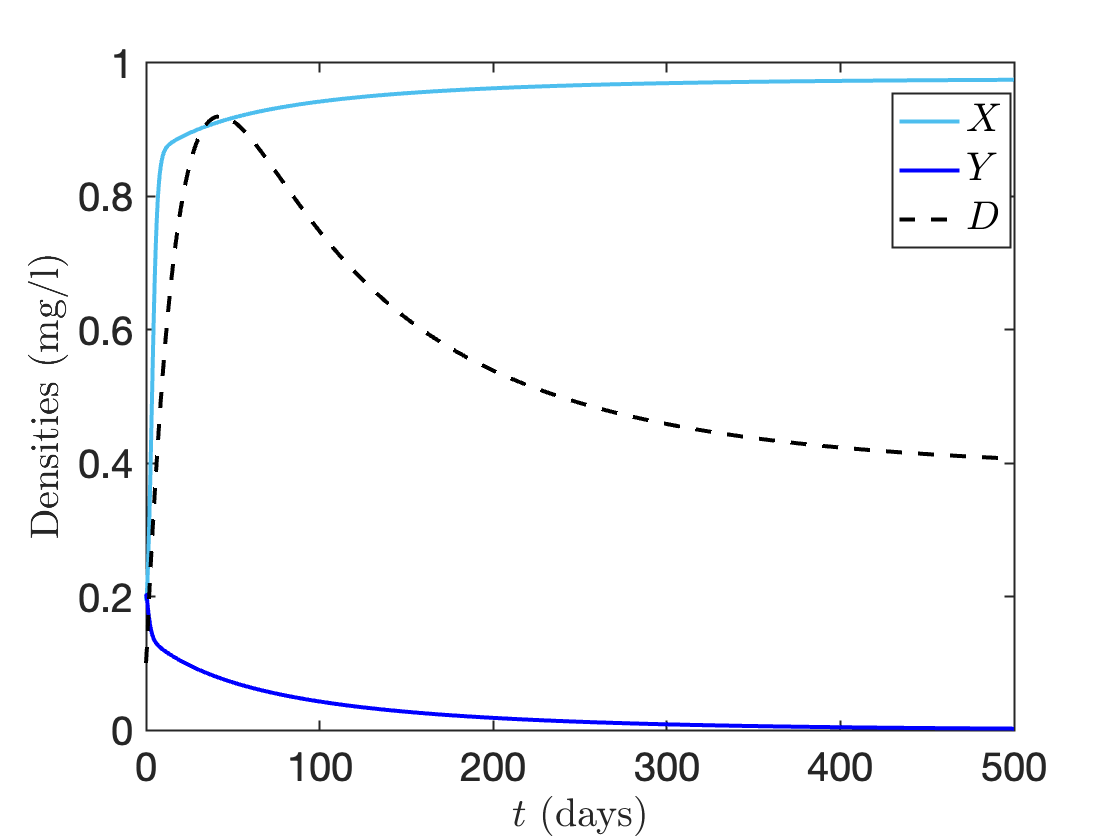}
      \includegraphics[width=0.3\textwidth]{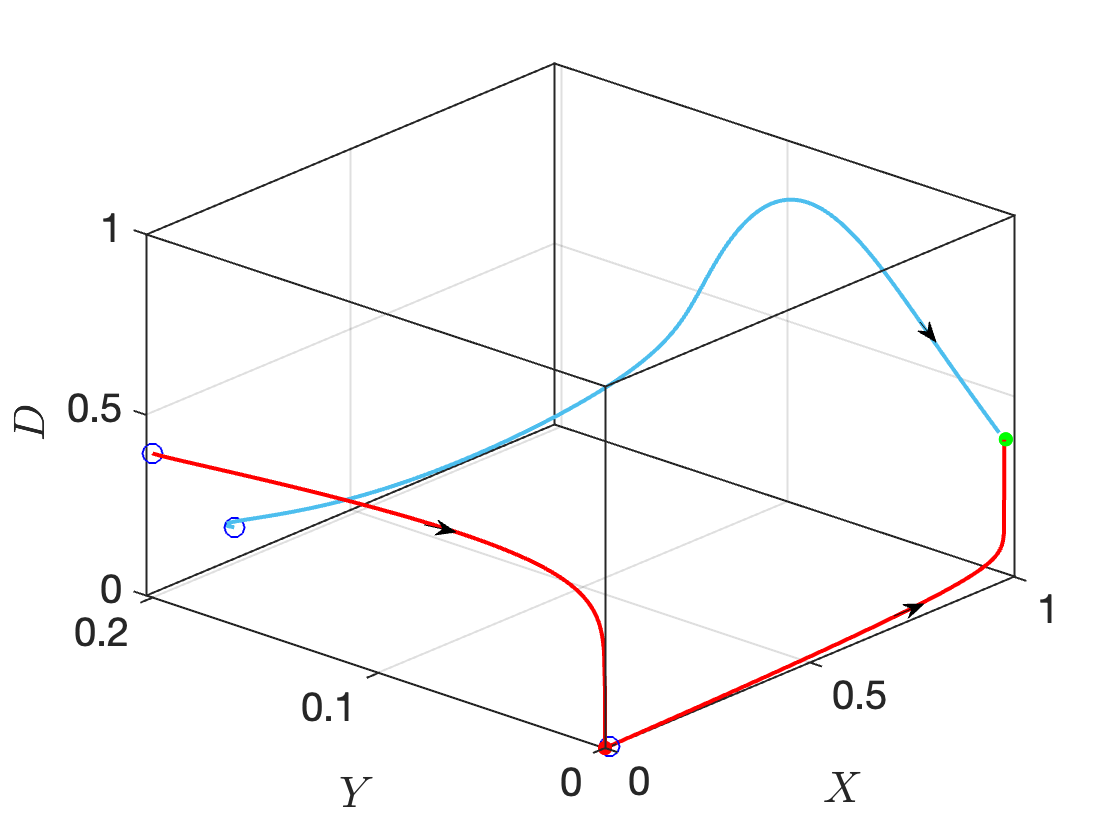}
      }

    \caption{Time series (left panel) and the trajectory (right panel) of the system \eqref{Eq:Model} for different $M_{\mathrm{in}}$. The stable limit cycle (magenta), stable equilibrium points (green), and unstable equilibria (red dots) are shown. The initial conditions (blue dots) and the trajectory (red) converging to $E_0$ and $E_1$ are presented. The parameter values are fixed at Table.~\eqref{tab:Model_para_table_aqua}.}
    \label{fig:Time_Series}
\end{figure}

\begin{figure}[ht!]
\centering
   \subfigure[]{\includegraphics[width=0.35\textwidth]{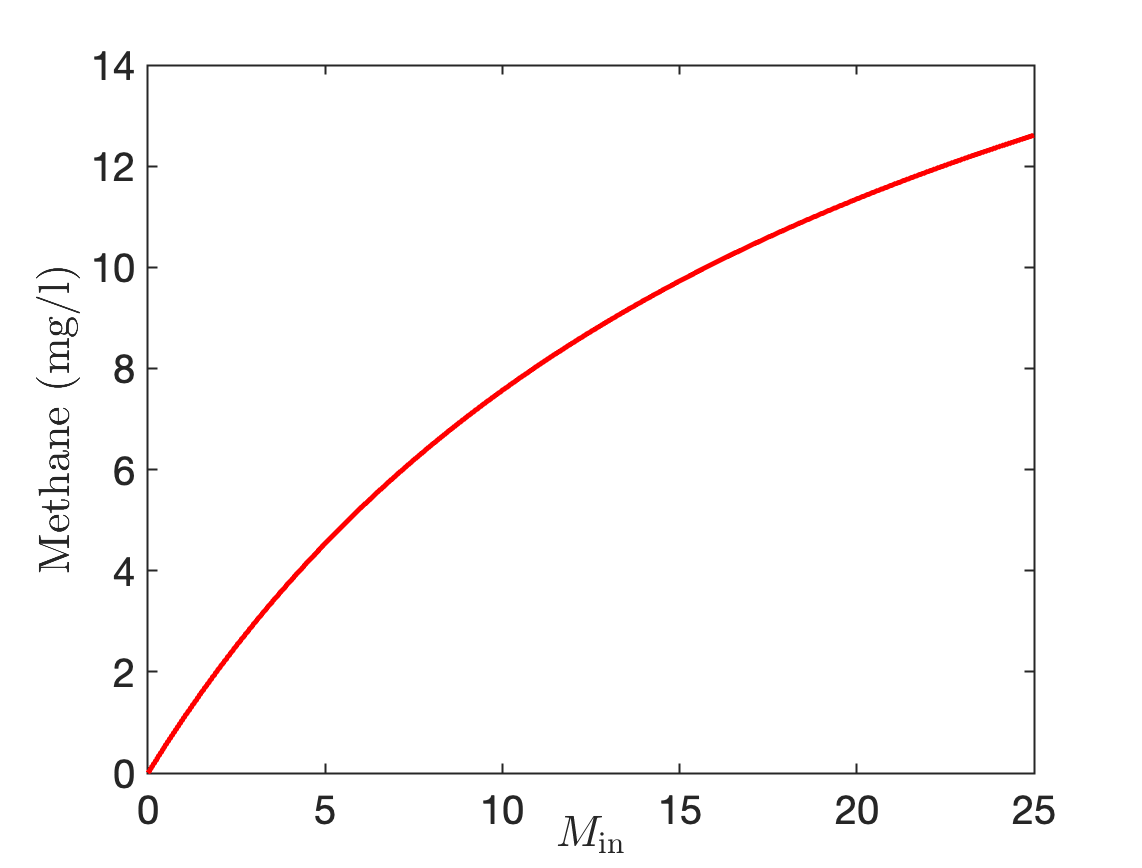}}
    \subfigure[]{\includegraphics[width=0.35\textwidth]{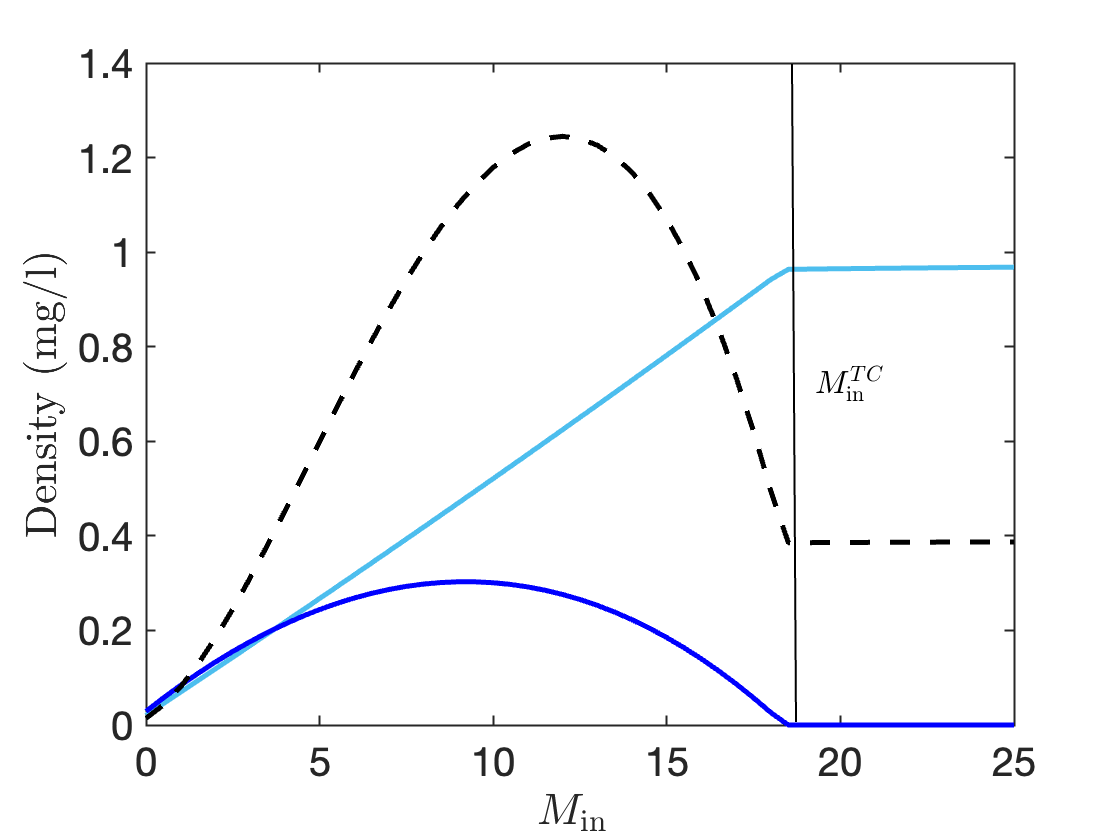}}
      \caption{The change in (a) methane concentration and (b) stable densities of algae (cyan), daphnia (blue) and detritus (black dotted) for varying rate of input $M_{\mathrm{in}}.$  $M_{\mathrm{in}}^{TC}$ represent the threshold for transcritical bifurcation. The other parameter values are fixed at Table.~\eqref{tab:Model_para_table_aqua}.}
    \label{fig:Bifurcation_Min_0_steady}
\end{figure}

\begin{figure}[ht!]
    \centering
   \subfigure[]{\includegraphics[width=0.35\textwidth]{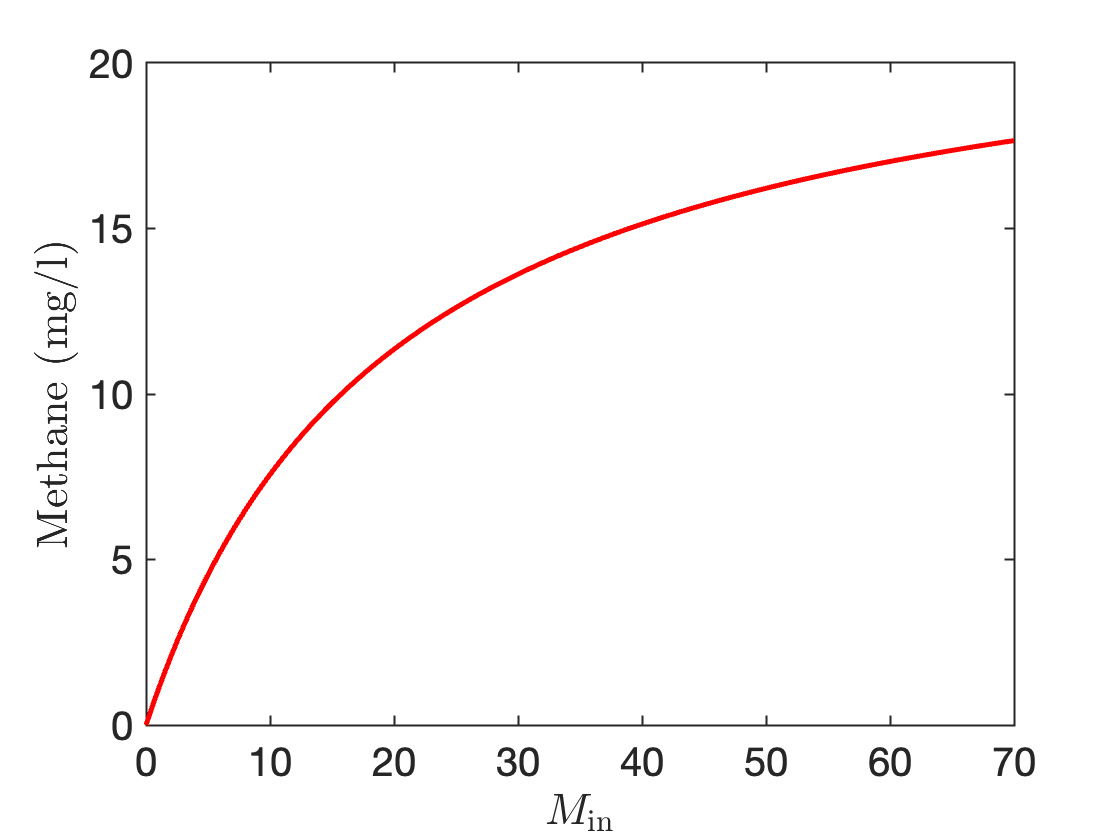}}
    \subfigure[]{\includegraphics[width=0.35\textwidth]{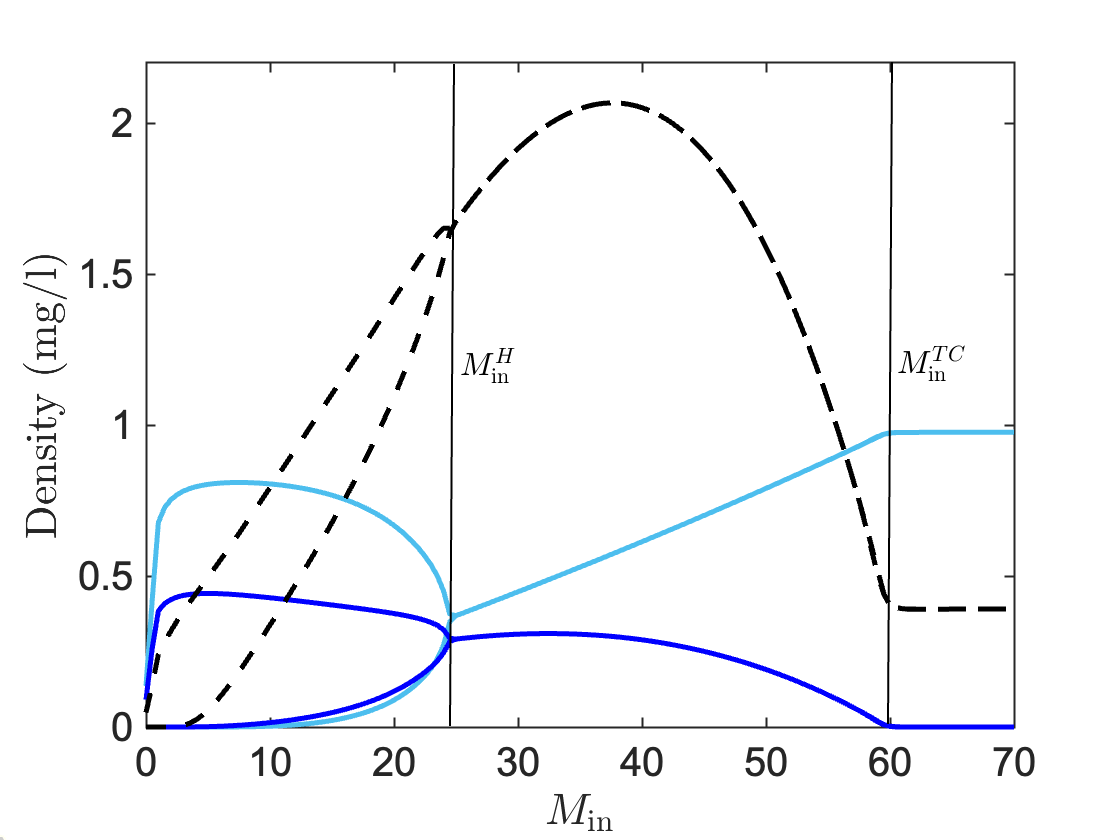}}
      \caption{The change in (a) methane concentration and (b)  densities of algae (cyan), daphnia (blue) and detritus (black dotted). Here, $M_{\mathrm{in}}^H$ and $M_{\mathrm{in}}^{TC}$ represent the thresholds for Hopf and transcritical bifurcation, respectively. The other parameter values are fixed at Table.~\eqref{tab:Model_para_table_aqua}.}
    \label{fig:Bifurcation_Min}
\end{figure}

At a low rate of methane input ($M_{\mathrm{in}} = 5$), the densities of algae and daphnia fluctuate between high and very low levels, spending significant time at low densities (cf. Fig.~\ref{fig:Time_Series}a). With \(M_{\mathrm{in}} = 15\), oscillations persist but with reduced amplitude (cf. Fig.~\ref{fig:Time_Series}b). The higher mean density of detritus reflects the increased mortality from elevated methane levels. As \(M_{\mathrm{in}}\) increases to 30, the unstable equilibrium stabilizes through enhanced algae growth, which temporarily supports daphnia populations (cf. Fig.~\ref{fig:Time_Series}c). However, at high methane input (\(M_{\mathrm{in}} = 65\)), algae reach their maximum capacity, while elevated methane levels simultaneously stress daphnia, suppressing reproduction and ultimately causing mortality (cf. Fig.~\ref{fig:Time_Series}d).

     


The bifurcation diagram (Fig.~\ref{fig:Bifurcation_Min}) further illustrates these dynamics. As methane input (\(M_{\mathrm{in}}\)) increases, the methane concentration in the water also rises (cf. Fig.~\ref{fig:Bifurcation_Min}(a)). At \(M_{\mathrm{in}} = 0\), the system exhibits a stable limit cycle, with the amplitude initially increasing for low \(M_{\mathrm{in}}\) values. The amplitude then decreases until \(M_{\mathrm{in}}^{H} = 24.5\), where a Hopf bifurcation stabilizes the system into a steady state with a declining daphnia population, although algal density continues to grow. A further critical transition occurs at \(M_{\mathrm{in}}^{TC} = 60\), where the zooplankton-free equilibrium is the sole stable equilibrium. While algal growth remains abundant, daphnia fails to graze and survive due to reduced movement, underscoring the chronic toxicity effects of high methane levels.

For a higher half-saturation constant (\(\beta = 0.9\)), predators require more resources to reach their growth potential. Under these conditions, the system remains stable across varying \(M_{\mathrm{in}}\) without periodic oscillations. The gradual changes in plankton population densities as \(M_{\mathrm{in}}\) increases are depicted in the bifurcation plot (Fig.~\ref{fig:Bifurcation_Min_0_steady}). At low to intermediate \(M_{\mathrm{in}}\) levels, daphnia experience positive feedback, with their density increasing due to ample food supply. However, as \(M_{\mathrm{in}}\) rises further, daphnia density starts to decline. At the critical threshold \(M_{\mathrm{in}}^{TC} = 18.5\), the system transitions to the stable zooplankton-free equilibrium.

\subsection{Fast methane dynamics induces long transients}

The intermediate subsystem \eqref{Eq:intermediate_model} along with parameter values 
\begin{equation}\label{par:parameter_values_nondimension}
    \begin{aligned}
        \theta = 22.7,\,&\sigma=0.8,\,\rho_1=0.4,\,
  \kappa=0.55,\,\mu_1=0.1,\,
  \mu_2=0.08,\,
  \rho_2=0.55,\\
  &\tau=6,\,\eta=0.01,\,
 \delta=40,\,\varepsilon_1=0.1,\,
  \varepsilon_2=0.0001,  
    \end{aligned}
\end{equation}
and $\zeta$ between $0-8,$ reveal critical insights into plankton behaviour in response to changes in methane concentration. We first consider the case where the system stabilizes at a coexistence equilibrium. The concentration of methane in water increases rapidly as a result of external input and the decomposition of organic matter, reaching a peak on a short timescale. Following this initial spike, phytoplankton density begins to grow steadily as methane dynamics slows down. This growth continues until the phytoplankton population reaches its maximum.

\begin{figure}[ht!]
    \centering
    \includegraphics[width=0.45\linewidth]{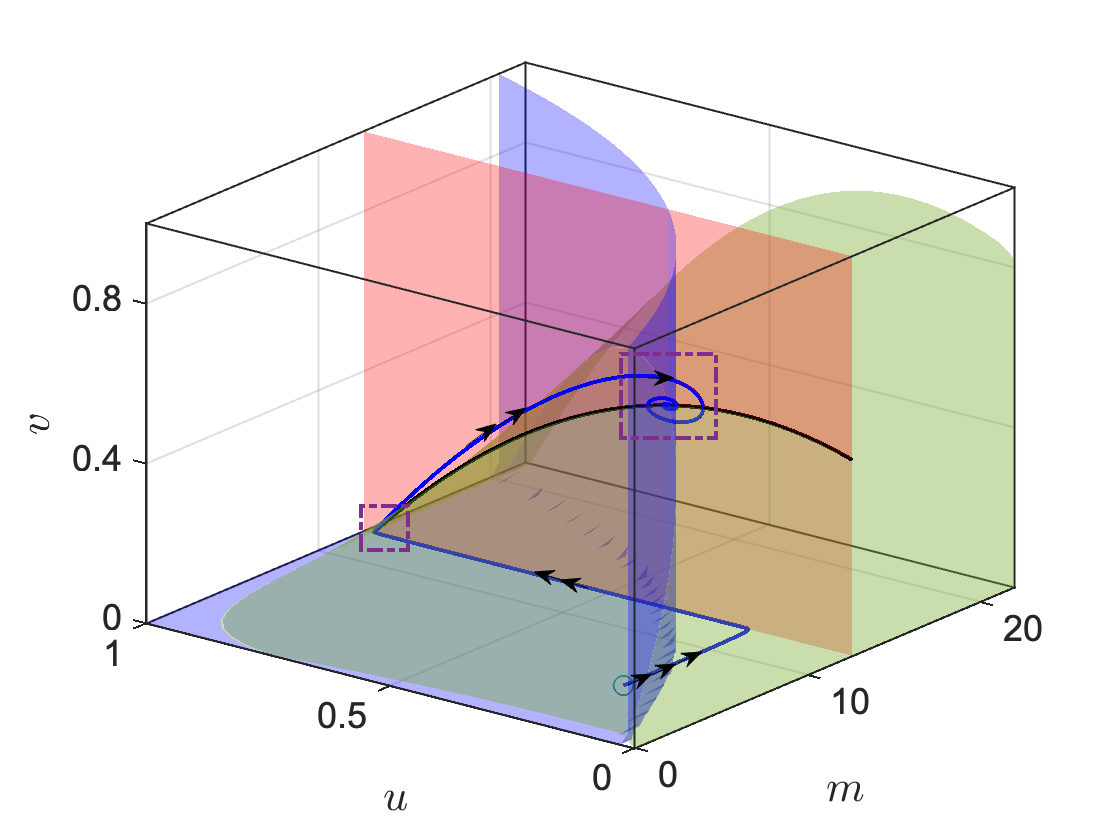}
    \caption{The system trajectory depicting transitions from fast methane dynamics to intermediate plankton dynamics. The critical manifold \(C^1\) in red, the manifold \(F_2(m, u, v) = 0\) in green, and the manifold \(F_3(m, u, v) = 0\) in blue for a fixed \(w_0\). Parameter values are fixed at \eqref{par:parameter_values_nondimension} with \(\zeta = 1\).}
    \label{fig:slow-inter-fast-steady}
\end{figure}

Starting from low zooplankton density, the system tends to approach a temporary equilibrium dominated by phytoplankton. During this phase, zooplankton density remains negligible for an extended period. This phenomenon is known as ``crawl-by transient". This transient highlights how zooplankton populations can experience delayed recovery when methane and phytoplankton dynamics dominate. Eventually, zooplankton density begins to increase as conditions improve, driven by available food resources from phytoplankton growth. However, as the zooplankton population expands, phytoplankton density starts to decline due to grazing pressure. This process unfolds over an intermediate timescale. The entire system gradually transitions toward a stable equilibrium, but the approach slows significantly as the trajectory nears the final steady state. Figure \ref{fig:slow-inter-fast-steady} illustrates these transitions, highlighting the rapid increase in methane concentration, the steady growth of phytoplankton, and the delayed recovery of zooplankton populations. 

The timescales of these processes depend on the relative speed of methane dynamics compared to plankton growth. For example, when methane dynamics are faster (smaller \(\varepsilon_1\)), zooplankton recovery is further delayed, extending the crawl-by transient \cite{Hastings18,Morozov_transient} period near the zooplankton-free equilibrium. Figure \ref{fig:time_Series_intermediate} shows how this delay manifests in the time series for different methane dynamics. Under different environmental conditions, the system exhibits oscillatory behaviour with a large amplitude, even at very low methane concentrations. These oscillations arise from alternating phases of gradual plankton decline and rapid growth. During the decline phase (slow flow), plankton densities decrease steadily and remain near negligible levels for an extended period. This is followed by a sudden resurgence of phytoplankton, driven by intermediate dynamics, which subsequently supports a rapid increase in zooplankton density. This pattern repeats, creating a dynamic balance between the two populations. Figure \ref{fig:slow-inter-fast-limitcycle}(a) illustrates such oscillatory behaviour. At high methane concentrations (\(m > 20\)), the system transitions to a zooplankton-free state, where algal densities remain high, but zooplankton populations collapse due to adverse conditions. This outcome highlights the sensitivity of zooplankton to elevated methane levels, which can disrupt the ecological balance (Fig.~\ref{fig:slow-inter-fast-limitcycle}(b)).

\begin{figure}[ht!]
    \centering
    \includegraphics[width=10cm, height=8cm]{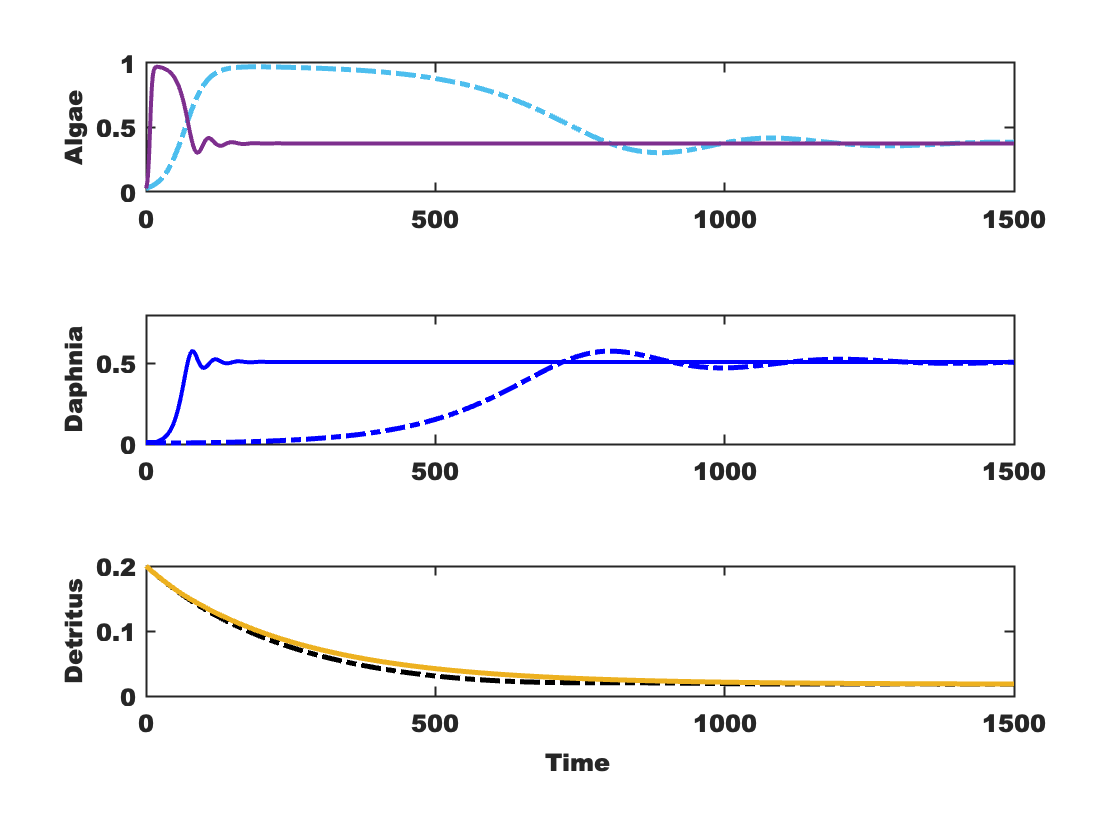}
    \caption{Time series of the trajectory for \(\varepsilon_1 = 0.1\) (solid curve) and \(\varepsilon_1 = 0.01\) (dashed curve), showing the prolonged crawl-by transient near the zooplankton-free state before approaching the steady state.}
    \label{fig:time_Series_intermediate}
\end{figure}

\begin{figure}[ht!]
    \centering
    \subfigure[]{\includegraphics[width=0.45\linewidth]{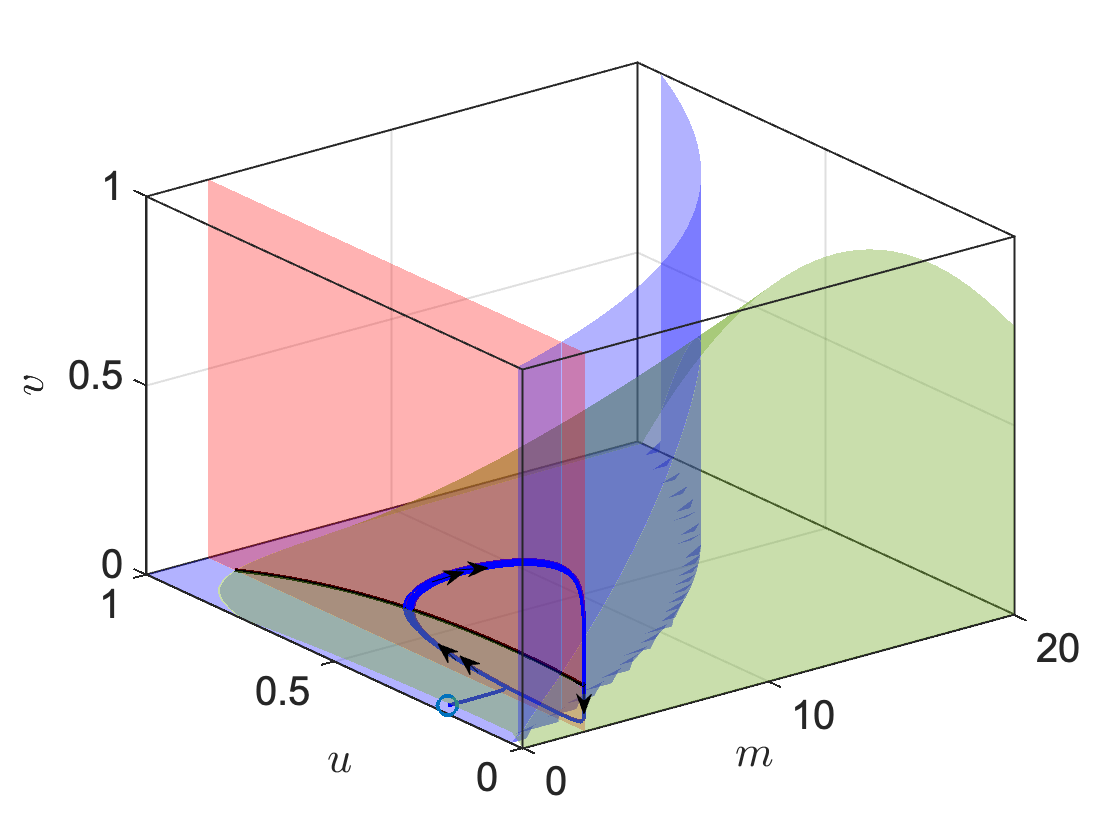}}
    \subfigure[]{\includegraphics[width=0.45\linewidth]{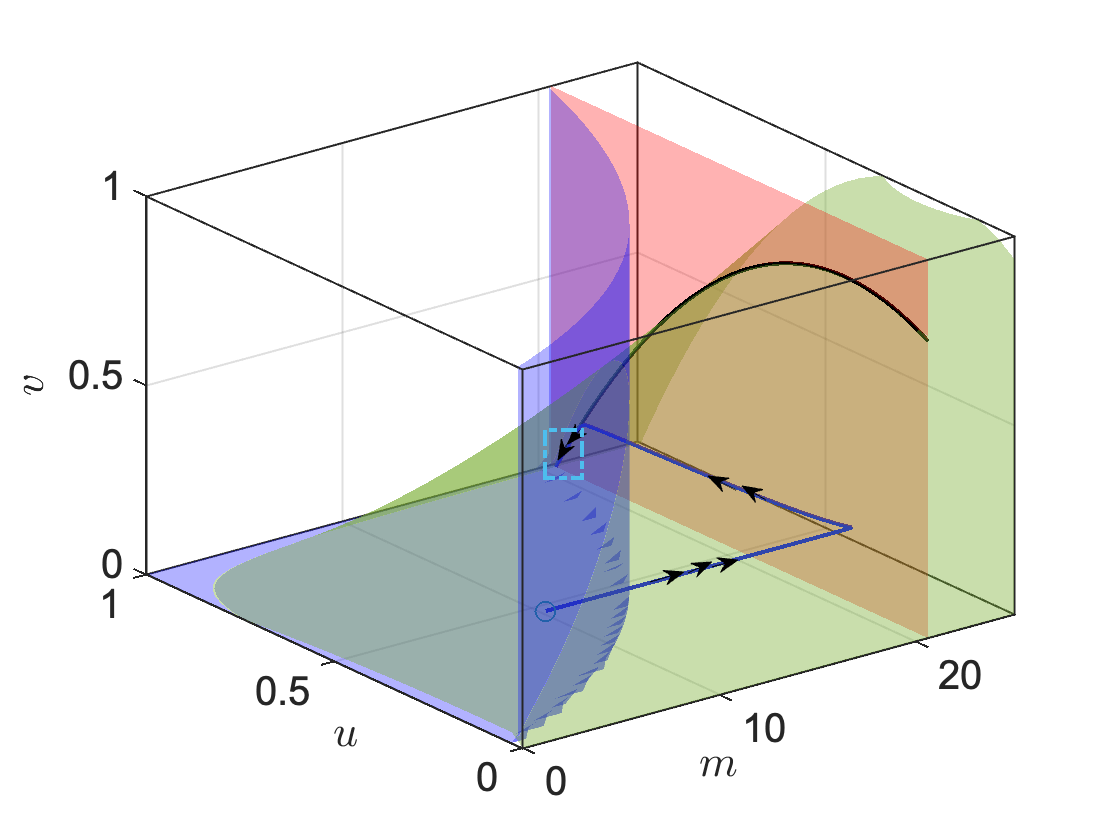}}
    \caption{(a) Intermediate-slow limit cycle in \(muv\) space. (b) System trajectory showing fast and intermediate transitions to the zooplankton-free state. Single, double, and triple arrows represent slow, intermediate, and fast flows, respectively.}
    \label{fig:slow-inter-fast-limitcycle}
\end{figure}

\subsection{Regime shift in terrestrial ecosystem}

The projections of the future atmospheric methane depend solely on the assumptions on the methane emissions from all the natural and man-made sources. However, the direct impact of the rise in methane concentration on terrestrial species has not yet been explored at full capacity. In this section, we provide a conceptual outlook on the variation of the species' density with the rise in the future methane concentration. As per the Shared Socioeconomic Pathway (SSP) scenarios based on climate modelling framework \cite{Kleinen}, the future methane predictions range between 1-8.5 ppm. Through our conceptual model, we capture this variation to observe the change in species density. The system \eqref{Eq:terrestrial_eq} is studied with parameter values fixed at \eqref{tab:Model_para_table_aqua} except $M_{\mathrm{out}}=0.00029,\,e_1=0.1,p=0.02,\,\gamma_M=0.4,\,\beta=0.76,\,e=0.35,\,d_X=0.004,\,d_Y=0.002,\,d_M=0.02.$ The choice of parameters and units are mentioned in SI \ref{sec:Parameterization}.

Our model shows a non-monotonic variation in the concentration of methane with $M_{\mathrm{in}}$ in the atmosphere. This is presented in the form of bifurcation diagrams in Fig.~\ref{fig:Bi_stable_dynamics}. The result shows that with an increase in $M_{\mathrm{in}},$ the model can encounter a bi-stability within a narrow region. The saddle-node bifurcation thresholds ($M_{\mathrm{in}}^{SN1}$ and $M_{\mathrm{in}}^{SN2}$) represent possible sudden shifts in the system depending on the parameter $M_{\mathrm{in}}.$
When $M_{\mathrm{in}}^{SN1}<M<M_{\mathrm{in}}^{SN2},$ the concentration of atmospheric methane ranges from $1.375<M<2.55$ and $M>4.89,$ respectively. We argue that within the range $1.375<M<2.55,$ the system shows an increase in the abundance of both the primary producers and consumers. This might create a positive illusion. However, with a further increase in the value of $M_{\mathrm{in}},$ if the parameter is pushed beyond the threshold $M_{\mathrm{in}}^{SN2},$ the system encounters an irreversible critical transition where the system dynamics shifts to an alternative steady state characterized by higher methane concentrations. This might increase the abundance of both the resource and consumer for a short time. However, if the system stays at this stable branch, with a further increase in methane concentration, the density of the consumer declines. The system encounters another transition point, the transcritical bifurcation, which marks the disappearance of the stable co-existence equilibrium. Beyond this critical threshold, the density of the consumer collapses and eventually leads to extinction.

\begin{figure}[ht!]
    \centering
    \subfigure[]{
    \includegraphics[width=0.45\linewidth]{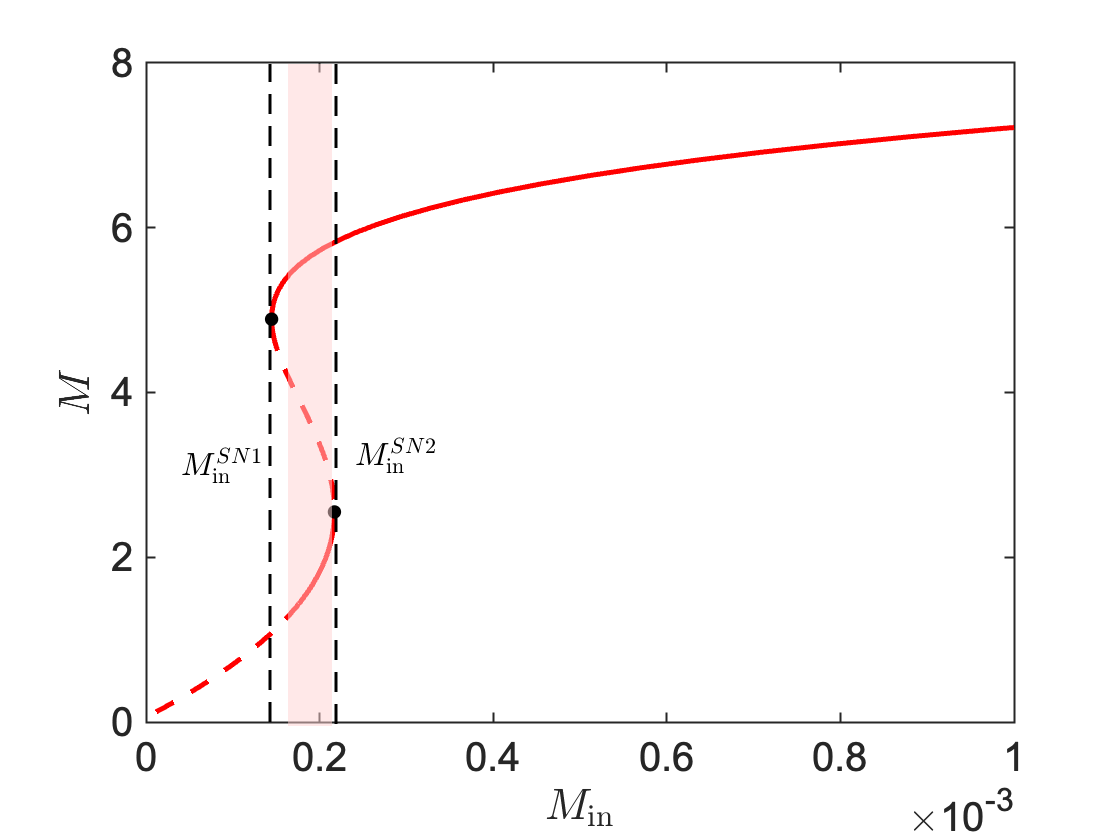}}
    \subfigure[]{
    \includegraphics[width=0.45\linewidth]{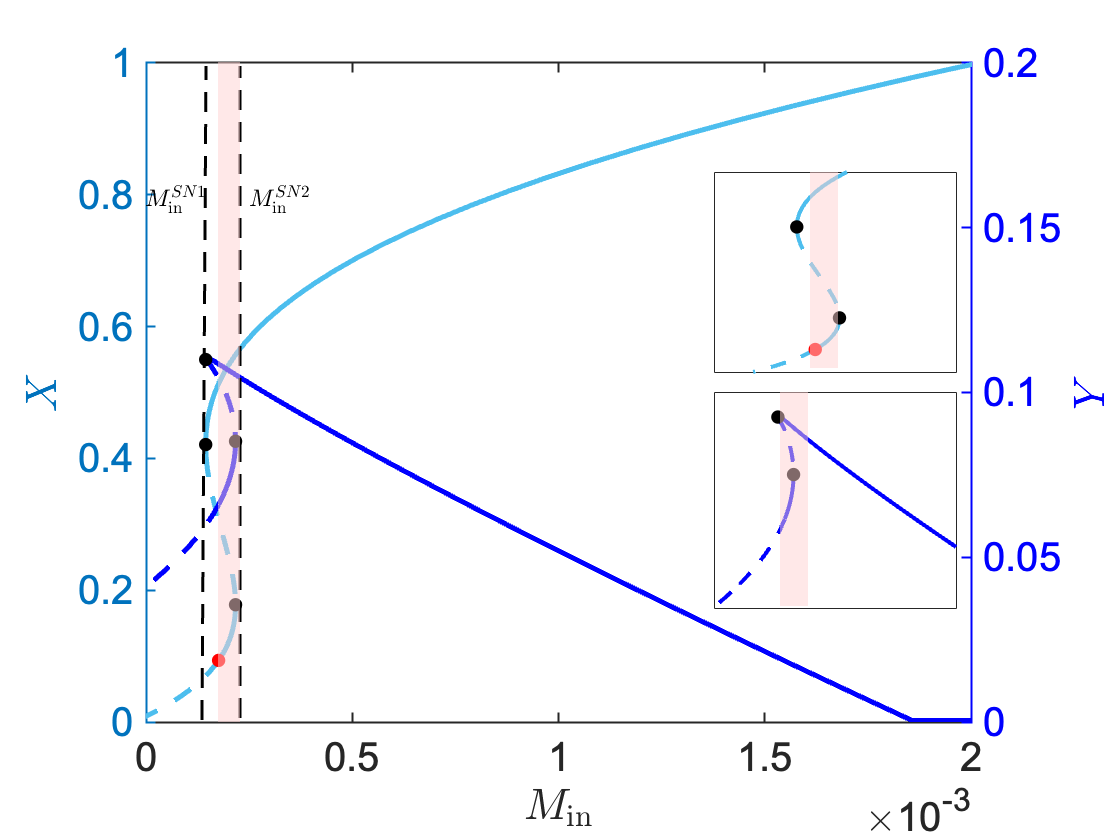}}
    \caption{(a) Change in methane concentration and (b) population densities for varying rates $M_{\mathrm{in}}$ showing the existence of multiple steady states within a narrow region.}
    \label{fig:Bi_stable_dynamics}
\end{figure}


\section{Discussion and Conclusion }\label{sec:Discussion}

Methane (\(\mathrm{CH}_4\)) is the second most significant greenhouse gas after carbon dioxide, playing a critical role in driving climate change \cite{howarth2011methane}. Major contributors of methane emission include the oil and gas industries, wetlands, landfills, and agricultural activities \cite{karakurt2012sources,Patin}. One of the most pressing consequences of rising methane levels is the global increase in surface temperatures \cite{gilbert2015climate}.  
Methane concentrations in the atmosphere have increased to two-and-a-half times their pre-industrial levels. Given the critical role of methane in climate dynamics, it is imperative to understand the factors and processes influenced by this upward trend. Addressing this knowledge gap is essential to effectively mitigate and manage its consequences. 

We developed a methane-resource-consumer-detritus model to offer a conceptual understanding of how increasing methane levels impact trophic dynamics and population structures in ecosystems. As an example, we primarily examined two ecosystems: terrestrial and aquatic. However, the framework can be extended to other ecosystems. Temperature is incorporated as a key factor influencing the growth and development of both producers and consumers, with growth peaking near their optimal temperature. Our findings reveal that species become increasingly vulnerable to rising temperatures in the presence of methane. Higher concentrations of methane significantly narrows the temperature range within which species can coexist. This suggests that as methane levels increase, the window for their survival near their optimal temperature diminishes. While primary producers may thrive at elevated temperatures, consumers face the risk of extinction due to environmental stress (cf. Fig.~\ref{fig:Two_parameter_T_Min}).

Through our analysis, we argue that methane can positively influence system dynamics by enhancing species growth, but only up to a certain threshold. For example, the densities of algae and Daphnia increase with rising methane input up to a critical limit. Biologically, this may create the illusion that methane contributes positively to ecosystem functioning. However, our results indicate that beyond a dissolved methane concentration of 4 mg/L, the population density of Daphnia declines, showing the lethal effects of elevated methane levels (cf. Fig.~\ref{fig:Bifurcation_Min_0_steady}). At high concentrations, methane-induced stress imposes several adverse effects on Daphnia, including reduced mobility for foraging, decreased hatching rates, impaired development, and increased mortality due to oxygen deficiency. These factors collectively weaken the population's resilience and long-term sustainability. Notably, our theoretical findings align with existing literature \cite{SafetyDataSheetSweden}, which identifies a chronic toxicity threshold for methane in Daphnia at approximately 3.88 mg/L.


Methane exerts a stabilizing effect on the system's dynamics (cf. Fig.~\ref{fig:Bifurcation_Min}). In the absence of methane, the resource-consumer system exhibits periodic oscillatory dynamics: when resource levels are high, consumer populations increase, which subsequently depletes the resource, leading to a decline in consumer populations. This cyclical pattern persists, resulting in continuous fluctuations. However, as methane levels increase, the oscillations diminish, and the system stabilizes, converging to a steady state where both resource and consumer densities remain constant. Further increase in methane, ultimately leading to the extinction of zooplankton. While toxin has a stabilizing effect on the dynamics of the system \eqref{Eq:Model}-\eqref{Eq:aquatic_eq}, it induces bistability in the terrestrial model \eqref{Eq:terrestrial_eq}. With the increasing trend of methane in the atmosphere, the ecosystem counters a catastrophic transition at $M_{\mathrm{in}}^{SN2}$ (cf. Fig.~\ref{fig:Bi_stable_dynamics}). Methane can be a slow threat to the ecosystem. 
Near this threshold, a slight increase in $M_{\mathrm{in}}$ shifts the system to an alternate steady state with elevated methane levels. At this stage, we observe an increase in resource density, whereas the density of consumers declines, eventually leading to extinction. The regime shifts in ecosystems due to the ongoing climate change have been a focus of research in ecology over the last few decades \cite{Scheffer,Petrovskii17}, but have rarely been explored with reference to toxin effect with an exception to few.


Section \ref{sec:Multiple timescales framework of the model} highlights the distinct timescales involved in the methane-organism-detritus interactions and their implications in real-world ecosystems. 
This approach effectively captures the transient behaviours within the aquatic ecosystem and can be extended to analyze terrestrial ecosystems as well.
Methane dynamics operate on a fast timescale, achieving equilibrium in the aquatic environment rapidly. In contrast, the interactions between plankton species evolve over an intermediate timescale. Meanwhile, the decomposition of dead organic matter occurs on a slow timescale, contributing gradually to methane production in the aquatic environment. 
Our results reveal that faster the accumulation of methane in water, longer is the duration of the crawl-by transient (cf. Fig~\ref{fig:time_Series_intermediate}), implying that the system requires a long time to recover from the zooplankton-free state, depicting local extinction. While low levels of methane may exhibit a transient positive effect on system dynamics, this process unfolds over an extended period. Moderate levels of dissolved methane, although not immediately lethal to Daphnia, can induce sub-lethal effects over time, impairing their survival and functional performance. These highlight the complex and time-dependent nature of methane's influence on ecosystem dynamics.

This study provides a mathematical framework as an initial step toward understanding methane's impact on ecosystems. To the best of our knowledge, no prior studies have been done in this area. However, further experimental and theoretical investigations are needed to unravel the complexities of species interactions influenced by rising methane levels. This research leaves many future research questions, including the potential behavioural adaptations and responses of species to changing climates. Investigating the distribution of species near the areas of high-emission zones. Further research in this direction is essential to assess methane's role as an environmental stressor and its broader implications for ecosystem stability and functioning.

\section*{Conflict of Interest}

The authors do not have any conflict of interest.

\section*{Acknowledgements}
The authors would like to acknowledge the Natural
Sciences and Engineering Research Council of Canada (NSERC) for the funding through NSERC
Alliance Missions grant (NSERC-AMG-577242) on anthropogenic greenhouse gas research.
H.W.'s research was partially supported by the Natural Sciences and Engineering Research Council of Canada (Individual Discovery Grant RGPIN-2020-03911 and Discovery Accelerator Supplement Grant RGPAS-2020-00090) and the Canada Research Chairs Program (Tier 1 Canada Research Chair Award).

\bibliographystyle{apalike2}


\clearpage
 {\LARGE   
\begin{center}
Supplementary Information for ``The silent threat of methane to ecosystems: Insights from mechanistic modelling"
\end{center}
 }
\begin{appendix}

\section{Parameter estimation}\label{sec:Parameterization}


\subsection{Optimal temperature }
Temperature is a critical factor in the functioning of biological species. Studies indicate that blue-green algae thrive in warmer water temperatures 
Consequently, we focus on the summer months to study plankton dynamics using our model. Data from the Lake Winnipeg shows a summer temperature gradient of \(15-25 \degree\)C \cite{LakeWinipeg}, while analyses of European and North American lakes report critical summer surface temperatures ranging from \(15.5-25 \degree\)C \footnote[1]{https://www.epa.gov/climate-indicators/great-lakes}. Thus, we consider the lake temperature \(T\) in our model to range between \(15-25 \degree\)C.

The optimal photosynthesis temperature for blue-green algae growth is reported as \(20-30 \degree\)C, with some studies suggesting a broader range of \(15-30 \degree\)C \cite{Singh15}. For Daphnia magna, the maximum growth rate occurs at \(20 \pm 2 \degree\)C \cite{Giebelhausen}. Considering these findings, we set the optimal temperatures in our model to \(T_X = 25 \degree\)C for algae and \(T_Y = 20 \degree\)C for daphnia.

\subsection{Methane saturation in water}
Methane reaches saturation in water at 22.7 mg/L under standard atmospheric pressure (1 atm) and temperature. According to Henry's law, the solubility of a gas in a liquid is proportional to the pressure \cite{Henry}, that is  
\[
\frac{\text{Solubility of gas at } P_1}{P_1} = \frac{\text{Solubility of gas at } P_2}{P_2}.
\]  

At depths of 0–10 m in a water body, the pressure is approximately 1 atm. Using Henry's law, the solubility of methane at this pressure is 22.7 mg/L. Therefore, we approximate \( \Bar{M} \) as 22.7 mg/L.

\subsection{Estimating \texorpdfstring{$\gamma_M$}{g}}

Methane indirectly affects algal growth through its interaction with methane-oxidizing bacteria. During methane oxidation, one mole of \(\mathrm{CH_4}\) reacts with two moles of $\mathrm{O_2}$ to produce \(\mathrm{CO_2}\). Experimental data on the direct influence of methane on algal growth is limited. A 20-day experiment in \cite{Enebo} demonstrated this effect, with data showing algal growth with and without methane (Fig.~\ref{fig:data_algae}). From this data, the growth rate due to methane was estimated between \(0.48\) and \(0.66\) mg/L per day. 

In the experiment, \(\mathrm{NaHCO_3}\) and \(\mathrm{CH_4}\) were added as carbon sources, with 50 mg/L of methane added daily and concentrations maintained at 65–70 mg/L. Using this data, the parameter \(\gamma_M\) is estimated as:  
\[
\gamma_M \approx \frac{\frac{\text{growth with methane}}{\text{growth without methane}}}{\text{methane concentration (mg/L)}} \approx 0.009\, \mathrm{mg^{-1} L}.
\]

However, experimental data on the direct influence of methane is scarce, and its effects vary across algal species and environmental conditions. To account for this variability, we consider a wide range of \(\gamma_M\) values in our model. Data from Fig. 2 in \cite{Hadiyanto} shows a linear relationship between \(\mathrm{CO_2}\) flow and microalgal growth (\textit{Chlamydomonas}), suggesting a maximum \(\gamma_M\) value of approximately 1 \(\mathrm{mg^{-1} L}\). Besides algae, it has been shown recently that methane can help in the growth and development of plants including seed germination, seedlings growth, and lateral rooting \cite{Wang}.
 Therefore, we use \( 0.009 \leq \gamma_M \leq 1 \) for our model.

\begin{figure}
    \centering
    \includegraphics[width=7cm,height=6cm]{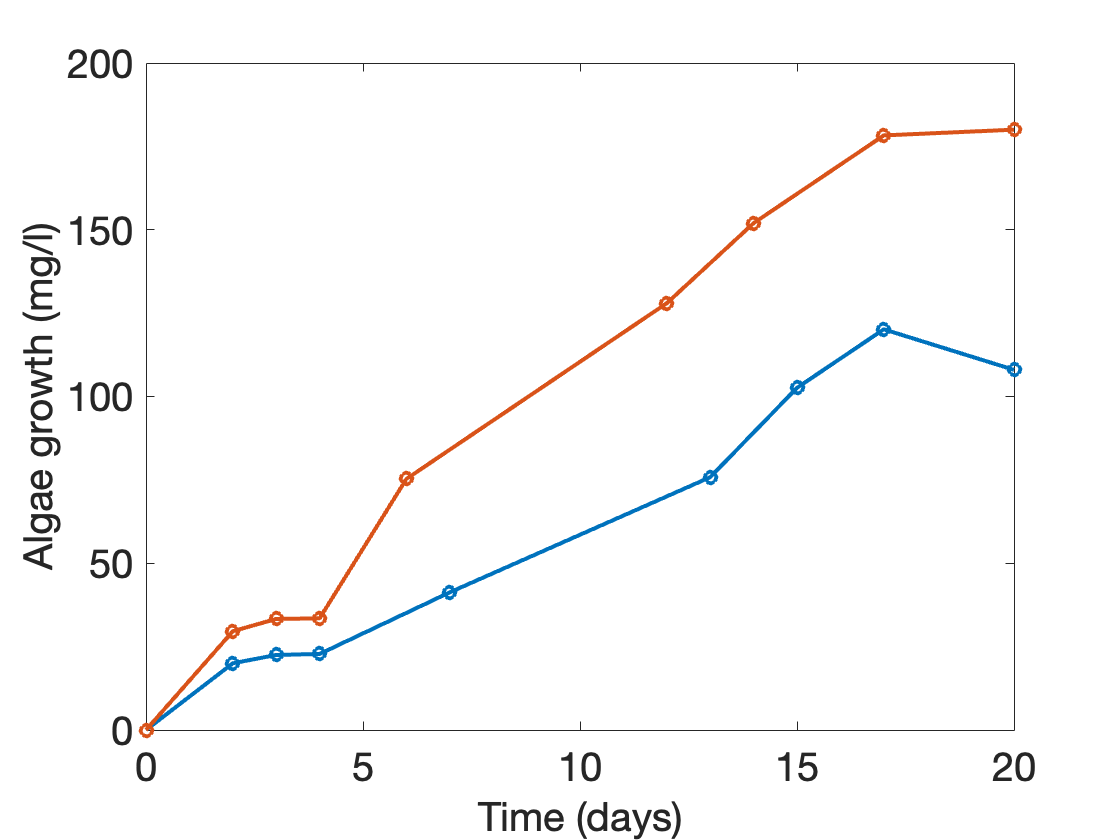}
    \caption{Data from \cite{Enebo} showing algal growth with methane (brown) and without methane (blue).}
    \label{fig:data_algae}
\end{figure}

\subsection{Estimating \texorpdfstring{$e$}{e}}
The maximum growth rate of daphnia in the presence of algae at an optimal temperature is \(0.3\, \mathrm{day}^{-1}\) \cite{Kankaala06}, corresponding to  
\[
\frac{e}{1+\gamma_1(T-T_Y)^2}\frac{\alpha X}{\beta+X} = 0.3.
\]  
At the optimal temperature \(T_Y = 20^\circ \mathrm{C}\) and \(X \to \infty\), this simplifies to \(e \alpha = 0.3\). With \(\alpha = 0.8\, \mathrm{day}^{-1}\) \cite{Chen17}, we find \(e = 0.375\).
Daphnia growth also depends on the quality of algae. Experiments \cite{DeMott} show juvenile growth rates of \(0.57-0.61\, \mathrm{day}^{-1}\) on phosphorous-rich algae. This gives us an estimate of $0.375<e<0.7625.$

\subsection{Estimating mortality }

Data on the direct toxicity of methane to zooplankton is quite limited. In oxic water layers, methane reduces dissolved oxygen which is necessary for the respiration of zooplankton. At \(25^\circ \mathrm{C}\), their average lifespan is up to 40 days\footnote[2]{https://animaldiversity.org/accounts/Daphnia\_magna/}, while at \(20^\circ \mathrm{C}\), they can live up to 3 months in the lab. This suggests a natural mortality rate \(d_Y \in [0.01, 0.025]\, \mathrm{day^{-1}}\).

The tests performed to measure the toxicity of gas on organisms are LC50 or EC50. LC50 (median lethal concentration) is the concentration of material (in this case, effluent) in water that is estimated to be lethal to $50 \%$ of the test organisms. Whereas, EC50 (median effective concentration) is the concentration of material in water estimated to cause a specified non-lethal or lethal effect in $50\%$ of test organisms. Based on Finland and Borealis, Sweden’s methane safety data sheet \cite{SafetyDataSheetFinland}, the 48-hour EC50 is \(69.4\, \mathrm{mg/L}\) and a 30-day chronic toxicity of \(3.886\, \mathrm{mg/L}\) \cite{SafetyDataSheetSweden}.
Using the Poisson process, we estimate mortality due to methane exposure. The percent mortality after \(t\) days is:  
\[
p_0(t) = 1 - e^{-(d_M  M) t},
\]  
where \(d_M  M\) is the methane-induced mortality. From the LC50 data, we estimate \(d_M \approx 0.0047\, \mathrm{mg^{-1} L^{-1} day^{-1}}\), and from the chronic toxicity data, \(d_M \approx 0.038\, \mathrm{mg^{-1} L^{-1} day^{-1}}\). Thus, we consider a range for \(d_M\) of \(0.0047-0.038\, \mathrm{mg^{-1} L^{-1} day^{-1}}\).

Data on the mortality of terrestrial species due to methane is scarce, as many species tolerate methane well and may also emit it metabolically. Smaller species at the base of the food chain, however, are more vulnerable due to their small body size. For example, the LC50 for rats is reported to exceed 50,000 ppm via direct inhalation \cite{SafetyDataRat}. To estimate the toxic effects of methane on smaller terrestrial species, we refer to allometric scaling \cite{Farrar}, which relates toxicity to body size through a non-linear relationship :
$$
T_2 = T_1\left(\frac{BW_2}{BW_1}\right)^b,
$$
where \(T_2\) is the extrapolated toxicity for Species 2, based on \(T_1\) (toxicity for Species 1), \(BW_1, BW_2\) are their respective body weights, and $b$ is the allometric scaling exponent. Based on this, we can approximately estimate methane-induced mortality (\(d_M\)) for various species. For our model, we assume a class of species with a body weight of less than 1 gram, for instance, the flies or bumblebees. Our choice of bumblebees is justified because it has been found that bumblebees are attracted to the methane leaks in the pipes, leading to death. Thus, assuming a weight of $0.1$g and using the above allometric scaling ($b=0.67)$, we estimate the lethal dose to be approximately 1044 ppm. Further, to incorporate the uncertainty factors while calculating the lethal toxicity, inter- and intra-specific adjustments are made \cite{Rusyn22}. A default value of 100 is used as proposed by WHO to account for the uncertainties in the extrapolation. Therefore, considering the uncertainty, the adjusted lethal dose is 10.44 ppm.  Using the Poisson process as before, we roughly estimate $0.008 < d_M < 0.79.$ Due to the unavailability of specific data, we rely mostly on the order of magnitude of estimates rather than the actual estimate. This is justified with our research question as rather than quantifying the direct effect we wanted to understand the direct effect through our modeling approach conceptually.

\subsection{Estimating \texorpdfstring{$M_{\mathrm{out}}$}{Mout}}

Methane loss occurs through the uptake of species, oxidation in oxic water or in the environment, and emission from the water surface to the atmosphere. In the aquatic ecosystem, due to the small size of daphnia and the lack of methane accumulation in organisms, methane uptake by daphnia is negligible and omitted from our model. The dominant loss is methane emission from the water surface to the atmosphere, which can be modeled using the air-water gas concentration gradient and the gas transfer velocity (\(k\)), given by \cite{Pajala23}
\[
F = k(C_w - C_{eq}),
\]
where \(F\) is the flux, \(C_w\) is the methane concentration in water, \(C_{eq}\) is the equilibrium concentration with methane partial pressure, and \(k\) is the gas transfer velocity (piston velocity). We estimate $
M_{\mathrm{out}} \approx \frac{k}{\text{depth (m)}}$.
Assuming a depth of 1 m (where ebullition probability is highest at \(0.8\) for depths of 0.5–1 m), data from 4 Swedish lakes show piston velocities ranging from \(0.08\) to \(2\,\mathrm{m\,day^{-1}}\) \cite{Pajala23}. Additional data from North American and Swedish lakes suggest a \(k\) ratio range of \(1\) to \(2.5\) \cite{Bastviken04}. Thus, we estimate \(M_{\mathrm{out}} \in [0.08, 2.5]\,\mathrm{day^{-1}}\). On the other hand, the methane that is being lost from water to the atmosphere gets trapped in the troposphere for a long time. The average lifespan of a methane molecule in the atmosphere is 9.6 years. So $M_{\mathrm{out}}$ = 0.00029 $\mathrm{day}^{-1}.$ We also consider the loss of methane from the troposphere to the stratosphere over a period of 1 year from the data presented in \cite{Ji20}. We further observe that the rate of loss of methane from the troposphere to the stratosphere is in order of $10^{-3}$ day$^{-1}.$

\subsection{Detritus related parameters}

Detritus accumulation results from the death of phytoplankton (\(d_X X\)) and zooplankton (\((d_Y + d_M M)Y\)). Decomposition rates of phytoplankton and plants reported in \cite{Enríquez} range from \(0.0001-0.2\, \mathrm{day^{-1}}\). Thus, we estimate the decomposition rate as \(p \in [0.0001, 0.2]\, \mathrm{day^{-1}}\).

Methane accumulation in detritus occurs along density gradients, where planktonic detritus can persist for weeks 
Experimental studies \cite{Kankaala} estimate methane emissions as \(4\%-25\%\) of the total net production for certain phytoplankton species, giving \(e_1 \in [0.04, 0.25]\).

\noindent The parameter estimates for the maximum growth rate ($r$), maximum carrying capacity of algae ($K$), the maximum predation rate ($\alpha$) and the half-saturation constant ($\beta$) are well studied in the literature \cite{Lürling,Chen17,Pranali24}. We mention the parameter ranges used in our study along with the references in Table ~\ref{tab:Model_para_table}.

\begin{table}[ht!]\scriptsize
\centering
    \begin{tabular}{ m{1cm}  m{8cm} m{3.5cm} m{1.5cm}} 
    \hline
 Para. & Description & Values &  Reference\\ 
 \hline 
  $T$ & Temperature of the water & 15-25 $\degree$C &\cite{LakeWinipeg} \\ 
  
  $T_X$ & Optimal temperature for the growth of algae & 25 $\degree$C & \cite{Lürling,Singh15}\\  

  $T_Y$ & Optimal temperature for the growth of daphnia & 20 $\degree$C & \cite{Giebelhausen}\\
  
  $c_1$ &  rate of dissolution of methane in water & 0.005-0.45 $\mathrm{mg^{-1} l}$ &  Assumed \\
  
  $\Bar{M}$ & Saturation of methane in water & 22.7 $\mathrm{mg l^{-1}}$& Estimated\\
  
  $M_{\mathrm{in}}$ & The external input of methane & $0-100\, \mathrm{mgl^{-1} day^{-1}}$& Assumed\\
  
  $M_{\mathrm{out}}$ & rate of loss of methane from lake to the atmosphere & $0.08-2.5\,
  \mathrm{day^{-1}}$ & Estimated\\
 
  $e_1$ & proportionality constant for formation of methane from detritus &  0-0.25 no unit & \cite{Kankaala} \\

  $r$ & maximum growth rate of algae & 0.93 $\pm$ 0.12 $\mathrm{day}^{-1}$ & \cite{Lürling}\\

  $\gamma_M$ & influence of methane to the growth of algae& 0.009-1 $\mathrm{mg^{-1} \, l}$ & Estimated\\

  $K$ & maximum carrying capacity & 0-10 $\mathrm{mg \, l}^{-1}$ &\cite{Chen17}\\
  
  $\alpha$ &maximum predation rate of prey& $0.48 - 1$ $\text{day}^{-1}$ & \cite{Chen17}\\
  
  $\beta$ & half-saturation constant& 0.25-0.95 $\mathrm{mg \, l}^{-1}$ &\cite{Chen17,Pranali24}\\
  
  $d_X$ & natural death rate of algae & 0.01-1 day$^{-1}$& \cite{Baker21}\\
  
  $e$ & maximal conversion efficiency of daphnia & 0.375-0.55 no unit & Estimated \\
  
  $d_Y$ & natural death rate of daphnia& 0.01-0.025 $\mathrm{day^{-1}}$ & Estimated \\
  
  $d_M$ & influence of methane on death of daphnia& 0.0047-0.038 $\mathrm{mg^{-1}l\,day^{-1} }$ & Estimated  \\
  
  $p$ & decomposition rate of detritus & 0.0001-0.2 $\mathrm{day^{-1}}$ & Estimated\\

  $\gamma_1$ & temperature related constant for growth of algae & 0-1 $\mathrm{\left(^oC\right)^{-2}}$ & Assumed\\

  $\gamma_2$ & temperature related constant for growth of daphnia & 0-1 $\mathrm{\left(^oC\right)^{-2}}$ &Assumed\\
  
 \hline
\end{tabular}
 \caption{Description of the parameters of the system (1).}
    \label{tab:Model_para_table}
\end{table}
For the simulation of the model (1)-(2), we fix the parameters at
\begin{equation}\label{par:parameter_temp}
    \begin{aligned}  
    &T_X=25,\, T_Y=20,\, c_1=0.05,\, \Bar{M}=22.7,\, M_{\mathrm{out}}=1,\, e_1=0.002,\, p=0.05, \\
    &\gamma_M=1,\, r=1.2,\, K=1,\, \alpha=0.8,\, \beta=0.25,\, e=0.55,\, d_X=0.02,\, d_Y=0.01,\, \\
    &d_M=0.02,\, \gamma_1=1,\, \gamma_2=1,
    \end{aligned}
\end{equation}
unless otherwise stated. For simulating the model (1) with (3) we fix the parameters at 
\begin{equation}\label{par:terres_parameter_temp}
    \begin{aligned}  &T=20,\,T_X=25,\,T_Y=20,  \,M_{\mathrm{out}}=0.00029,\,e_1=0.1,p=0.02,\\ &\gamma_M=0.4,\,
r=1.2,\,K=1,\,\alpha=0.8,\,\beta=0.76,\,e=0.35,\,d_X=0.004,\,\\
&d_Y=0.002,\,d_M=0.02,\gamma_1=1,\,\gamma_2=1,
    \end{aligned}
\end{equation}
and consider $M_{\mathrm{in}}$ as bifurcation parameter. The authors in \cite{Kleinen} studied various scenarios of the future where the atmospheric concentration of methane can range from 1-8.5 ppm. We assume $M_{\mathrm{in}}$ from the growth rate of emission of methane. For the period of 40 years, starting from 1984 till date, the growth in the concentration is approximately 0.3 ppm. Thus we assume $M_{\mathrm{in}}$ for the terrestrial model in the range of $10^{-5}$ to $10^{-3}$ ppm per day. 

\section{Mathematical analysis} \label{App:Mathematical_Analysis}

The well-posedness of the model is proved in the following theorem by demonstrating the positivity and boundedness of the solutions. Let 
\begin{equation*}
    \Omega_a = \bigl\{(M, X, Y, D) : 0 \leq M \leq \Bar{M}, 0 \leq X \leq K, Y \geq 0, D \geq 0 \bigr\}.
\end{equation*}  
and
\begin{equation*}
    \Omega_t = \bigl\{(M, X, Y, D) : M \geq 0, 0 \leq X \leq K, Y \geq 0, D \geq 0 \bigr\}.
\end{equation*}

\begin{theorem}
\label{th:invariant_set}
    The solutions of system (1)-(3), given initial conditions within the set \( \Omega_a \) (or \( \Omega_t \)), remain in \( \Omega_a \) (or \( \Omega_t \)) for all forward time.
\end{theorem}
\begin{proof}
\label{proof:invariant_set}
Assume $S(t)=(M(t), X(t),Y(t), D(t))$ is a solution of system (1) with (2) or (3) with $S(0)\in \Omega_t$ and $t_1$ is the first time that $S(t)$ touches or crosses the boundary of $\Omega_t$. We prove the theorem by contradiction in the following cases

\textit{Case 1:} \( X(t_1) = 0 \). For all \( t \in [0, t_1] \), we have \( \frac{\mathrm{d}X}{\mathrm{d}t} \geq - d_X X \). Thus, \( X(t_1) \geq X(0) e^{-d_X t_1} > 0 \), which contradicts \( X(t_1) = 0 \). Therefore, \( S(t_1) \) cannot reach this boundary.

\textit{Case 2:} \( Y(t_1) = 0 \).
For all \( t \in [0, t_1] \), \( \frac{\mathrm{d}Y}{\mathrm{d}t} \geq - d_Y Y \). Therefore, \( Y(t_1) \geq Y(0) e^{-d_Y t_1} > 0 \), which contradicts \( Y(t_1) = 0 \). Hence, \( S(t_1) \) cannot reach this boundary.

\textit{Case 3:} \( X(t_1) = K \).
For all \( t \in [0, t_1] \),
\[
\frac{\mathrm{d}X}{\mathrm{d}t}\Big|_{t = t_1} = - \frac{\alpha K Y}{\beta + K} - d_X K \leq 0.
\]
Thus, \( S(t_1) \) cannot cross this boundary.

\textit{Case 4:} \( M(t_1) = 0 \).
For all \( t \in [0, t_1] \), \( \frac{\mathrm{d}M}{\mathrm{d}t} \geq - M_\text{out} M \), so \( M(t_1) \geq M(0) e^{- M_\text{out} t_1} > 0 \), which contradicts \( M(t_1) = 0 \). Therefore, \( S(t_1) \) cannot reach this boundary.

\textit{Case 5:} \( D(t_1) = 0 \).
For all \( t \in [0, t_1] \), \( \frac{\mathrm{d} D}{\mathrm{d}t} \geq - p D \), so \( D(t_1) \geq D(0) e^{- p t_1} > 0 \), which contradicts \( D(t_1) = 0 \). Therefore, \( S(t_1) \) cannot reach this boundary.

If, additionally, \( S(t) \) is a solution of system (1) with (2), and \( S(0) \in \Omega_a \), let \( t_1 \) be the first time that \( S(t) \) touches or crosses the boundary of \( \Omega_a \). Consider the following case:

\textit{Case 6:} \( M(t_1) = \Bar{M} \).
We have \( \frac{\mathrm{d}M}{\mathrm{d}t} \Big|_{t=t_1} = - M_\text{out} \bar{M} \leq 0 \), implying that \( S(t_1) \) cannot cross this boundary.

In summary, the solution \( S(t) \) of system (1), starting within \( \Omega_t \) (or \( \Omega_a \)), will remain in \( \Omega_t \) (or \( \Omega_a \)) for all forward time. This completes the proof.  

\end{proof}

This theorem guarantees that the solutions are always positive and that the producer density is bounded. That is, the model for both terrestrial and aquatic ecosystems remains biologically meaningful.

\subsection{Steady States of the system}\label{App:Steady_state}
To analyze the asymptotic behavior of the system, we study the existence and stability of equilibria. By direct calculations, we determined that the system (1)-(3) can exhibit three types of equilibrium points:

1. Total Extinction State of the form \(E_0(M^*, 0, 0, 0)\), where resource, consumer and detritus are extinct is given by
$$ E_0^a = \left(\frac{c_1 \Bar{M} M_\text{in}}{M_\text{out} + c_1 M_\text{in}}, 0, 0, 0\right),\,\,E_0^t = \left(\frac{M_\text{in}}{M_\text{out}}, 0, 0, 0\right),$$
where the superscript $`a$' and $`t$' correspond to the model for aquatic ecosystems and terrestrial ecosystems, respectively.

2. The semi-trivial Equilibrium \(E_1(M^*, X^*, 0, D^*)\) representing the extinction of consumer population is given by
$$E_1^a = \left(\frac{c_1 \Bar{M} \left(M_\text{in} + e_1 d_X X^* \right)}{M_\text{out} + c_1 \left(M_\text{in} + e_1 d_X X^* \right)}, X^*, 0, \frac{d_X X^*}{p}\right),\,\,E_1^t = \left(\frac{M_\text{in} + e_1 d_X X^*}{M_\text{out}}, X^*, 0, \frac{d_X X^*}{p}\right).$$
The description of the superscript is the same as explained before.

3. The coexistence equilibrium of the form \(E_*(M^*, X^*, Y^*, D^*)\), where both resource and consumer coexist is given by  
   $$
   M^* = \frac{1}{d_M}\left(\frac{e}{1 + \gamma_2\left(T - T_Y\right)^2} \frac{\alpha X^*}{\beta + X^*} - d_Y\right), \quad    
   Y^* = \frac{p D^* - d_X X^*}{d_Y + d_M M^*}.
  $$ 
  However, the equilibrium of the detritus density for the two models is obtained as
  $$D^{*a} = \frac{M_\text{out} M^*}{c_1 e_1 p \left(\Bar{M} - M^*\right)} - \frac{M_\text{in}}{e_1 p},\,\,D^{*t} = \frac{M_\text{out} M^* - M_\text{in}}{e_1 p}.$$

\subsection{Linear stability}\label{Appendix:Stability}
We present the following stability results for the equilibria described above. Re-writing the system (1) in the form of $\frac{\mathrm{d} U}{\mathrm{d} t} = \Phi(U)$.

\subsubsection{Stability of trivial-equilibrium}
\begin{theorem}
\label{Th: extinction E0 stability}
    The extinction equilibrium $E_0^a(M^*,0,0,0)$ ($E_0^t$) is locally asymptotically stable (LAS) if $ \frac{c_1M_{\mathrm{in}}}{M_{\mathrm{out}}+c_1 M_{\mathrm{in}}}<\frac{d_X \gamma_1 \left(T-T_X\right)^2 + d_X - r}{\bar{M}r \gamma_M}$, otherwise, it is unstable.
\end{theorem}

\begin{proof}
\label{proof: extinction E0 stability}
At \( E_0^a \), the Jacobian matrix is given by
\begin{equation}
D \Phi (E_0^a) =
    \begin{pmatrix}
       -c_1 M_{\text{in}} - M_{\text{out}} & 0 & 0 & c_1 e_1 p \max\{\Bar{M}-M^*,0\} \\
       0 & \frac{r \left(1 + \gamma_M M^*\right)}{1 + \gamma_1 \left(T - T_X\right)^2}-d_X & 0 & 0 \\
       0 & 0 & -d_Y - d_M M^* & 0 \\
       0 & d_X & d_Y + d_M M^* & -p
    \end{pmatrix}.
\end{equation}
The eigenvalues of the Jacobian matrix \( D\Phi (E_0^a)\) are $-c_1 M_{\text{in}} - M_{\text{out}} <0 $, $\frac{r \left(1 + \gamma_M M^*\right)}{1 + \gamma_1 \left(T - T_X\right)^2}-d_X$, $-d_Y - d_M M^*<0$, $-p<0$. Therefore, \( E_0^a \) is LAS if and only if $\frac{c_1M_{\mathrm{in}}}{M_{\mathrm{out}}+c_1 M_{\mathrm{in}}}<\frac{d_X \gamma_1 \left(T-T_X\right)^2 + d_X - r}{\bar{M}r \gamma_M},$
otherwise, \( E_0^a \) is an unstable saddle node.
The same argument applies to \( E_0^t \), completing the proof.
\end{proof}
This implies that if the growth rate influenced by methane is less than the death rate of the prey, neither the prey nor the predator can survive. Moreover, for the semi-trivial equilibrium, where only the prey survives, the stability is analyzed in Theorem \ref{Th: prey-only E1 stability}. The coexistence state of both prey and predator can also be either LAS or unstable, with the conditions detailed in Theorem \ref{Th:coexistence} and Remark \ref{rem: sufficient and necessary conditions for E_*}.

\subsubsection{Stability of semi-trivial equilibrium}

\begin{theorem}
\label{Th: prey-only E1 stability}
The semi-trivial equilibrium \(E_1^a\) (\(E_1^t\)) is LAS for the aquatic (terrestrial) system if the inequalities \eqref{eq: condition for E1 stable} hold with \(a_1^a, a_2^a\) (\(a_1^t, a_2^t\)). Otherwise, it is unstable.
\end{theorem}
\begin{proof}
\label{proof:prey-only E1 stability}

At \(E_1^a\), the Jacobian matrix \(D \Phi(E_1^a)\) is
\[\scriptsize
\begin{pmatrix}
       -c_1(M_{\text{in}} + e_1 p D^*) - M_{\text{out}} & 0 & 0 & c_1 e_1 p \max\{\Bar{M} - M^*, 0\} \\
       \frac{r \gamma_M X^* \left(1 - \frac{X^*}{K}\right)}{1 + \gamma_1 (T - T_X)^2} & \frac{r(1 + \gamma_M M^*) }{1 + \gamma_1 (T - T_X)^2}\left(1 - \frac{2X^*}{K}\right) - d_X & -\frac{\alpha X^*}{\beta + X^*} & 0 \\
       0 & 0 & \frac{e \alpha X^*}{(1 + \gamma_2 (T - T_Y)^2)(\beta + X^*)} - (d_Y + d_M M^*) & 0 \\
       0 & d_X & d_Y + d_M M^* & -p
\end{pmatrix}.
\]
Simplifying notation, becomes
\[
D \Phi(E_1^a) =
\begin{pmatrix}
    -a_1^a & 0 & 0 & a_2^a \\
    a_3 & -a_4 & -a_5 & 0 \\
    0 & 0 & -a_6 & 0 \\
    0 & d_X & a_7 & -p
\end{pmatrix}.
\]
The characteristic equation is
\[
\left(-a_6 - \lambda\right)\left((-a_1^a - \lambda)(-a_4 - \lambda)(-p - \lambda) + a_2^a a_3 d_X\right) = 0.
\]
Therefore, one eigenvalue is given by $-a_6$. The equilibrium $E_1$ is stable if and only if all real parts of eigenvalues of $D\Psi(E_1^a)$ are negative. Therefore, we require $a_6 >0$. To make sure all the rest eigenvalues are negative, we find the Hurwitz matrix for 
the rest three eigenvalues. The Hurwitz matrix is given by
\[
\text{Hurwitz matrix} =
\begin{pmatrix}
    b_1 & b_3 & 0 \\
    b_0 & b_2 & 0 \\
    0 & b_1 & b_3
\end{pmatrix},
\]
where
\[
\begin{aligned}
    b_0 &= 1, \quad b_1 = a_1^a + a_4 + p, \\
    b_2 &= a_1^a a_4 + a_1^a p + a_4 p, \quad b_3 = a_1^a a_4 p - a_2^a a_3 d_X.
\end{aligned}
\]
Therefore, we need $b_1>0$, $b_1 b_2 - b_0 b_3 >0$, $b_3>0$.
To summarize, $E_1^a$ is stable if and only if 
\begin{equation}
\label{eq: condition for E1 stable}
    \begin{aligned}
        & a_6 > 0, \quad a_1^a + a_4 + p > 0, \quad a_1^a a_4 p > a_2^a a_3 d_X, \\
        & (a_1^a + a_4 + p)(a_1^a a_4 + a_1^a p + a_4 p) > a_1^a a_4 p - a_2^a a_3 d_X.
    \end{aligned}
\end{equation}
holds, otherwise, it's unstable. For the terrestrial case, we use the same argument by substituting $a_1^a$ and $a_2^a$ by $a_1^t$ and $a_2^t$.
\end{proof}
\subsubsection{Stability of interior equilibrium}
\begin{theorem}
\label{Th:coexistence}
The interior equilibrium \(E_*^a\) is LAS for the aquatic system if \eqref{eq:E_*_stability_1}–\eqref{eq:E_*_stability_4} hold, and \(E_*^t\) is LAS for the terrestrial system if \eqref{eq:E_*_stability_2}–\eqref{eq:E_*_stability_5} hold. 
\end{theorem}

\begin{proof}
\label{proof:Th:coexistence}
Given the complexity of the system (1), we derive sufficient conditions for the stability of $E_*$. The Jacobian matrix $D\Phi (E_*)$ is given by
   \begin{equation}
\tiny
\begin{pmatrix}
       -c_1(M_{\text{in}} + e_1 p D^*) - M_{\text{out}} & 0 & 0 & c_1 e_1 p \max\{\Bar{M} - M^*, 0\} \\
       \frac{r \gamma_M X^* \left(1 - \frac{X^*}{K}\right)}{1 + \gamma_1 (T - T_X)^2} & \frac{r\left(1+\gamma_MM^*\right)\left(K-2X^*\right)}{K\left(1 + \gamma_1 (T - T_X)^2\right)} - d_X-\frac{\alpha\beta Y^*}{\left(X^*+\beta\right)^2} & -\frac{\alpha X^*}{\beta + X^*} & 0 \\
       -d_M Y^* & \frac{e\alpha\beta Y^*}{\left(1+\gamma_2(T-T_Y)^2\right)(X^*+\beta)^2} & \frac{e \alpha X^*}{(1 + \gamma_2 (T - T_Y)^2)(\beta + X^*)} - (d_Y + d_M M^*) & 0 \\
       d_M Y^* & d_X & d_Y + d_M M^* & -p
    \end{pmatrix}.
\end{equation}
The Gershgorin circle theorem \cite{Varga} states that all eigenvalues of \(D\Phi(E_*)\) lie within the Gershgorin discs \(C_i(a_{ii}, R_i)\), where \(a_{ii}\) is the diagonal entry, and \(R_i = \sum_{j \neq i} |a_{ij}|\) is the sum of the absolute values of the off-diagonal entries. For stability, all discs must lie in the left half-plane. This yields the following conditions
\begin{equation}\label{eq:E_*_stability_1}
    c_1 e_1 p \max\{\Bar{M} - M^*, 0\} < c_1(M_{\text{in}} + e_1 p D^*) + M_{\text{out}} ,
\end{equation}

\begin{equation}\label{eq:E_*_stability_2}
\begin{aligned}
    &\frac{r \left(1 + \gamma_M M^*\right)\left(K-2X^*\right) + r \gamma_M X^* \left(K - X^*\right)}{K\left(1 + \gamma_1 (T-T_X)^2\right)} < d_X + \frac{\alpha\beta( Y^* - X^*) - \alpha {X^*}^2 }{\left( \beta + X^*\right)^2},
\end{aligned}    
\end{equation}

\begin{equation}\label{eq:E_*_stability_3}
    \begin{aligned}
       \frac{e \alpha (X^{*2} + X^* \beta + Y^* \beta)}{(1 + \gamma_2 (T - T_Y)^2) (X^* + \beta)^2} < d_Y + d_M (M^* - Y^*),
    \end{aligned}
\end{equation}

\begin{equation}\label{eq:E_*_stability_4}
    \begin{aligned}
       d_X + d_Y + d_M (M^* + Y^*) < p.
    \end{aligned}
\end{equation}
The inequalities \eqref{eq:E_*_stability_1}–\eqref{eq:E_*_stability_4} ensure that all Gershgorin discs lie entirely in the left half-plane, which implies that the real parts of all eigenvalues of \(D\Phi(E_*^a)\) are negative. Thus, \(E_*^a\) is locally asymptotically stable.
For the terrestrial system, condition \eqref{eq:E_*_stability_1} is replaced by
\begin{equation} 
\label{eq:E_*_stability_5}
e_1 p < M_{\text{out}}. 
\end{equation} \end{proof}

\begin{remark}
\label{rem: sufficient and necessary conditions for E_*}
We investigate the sufficient and necessary conditions for the local asymptotic stability of the interior equilibria \(E_*^a\) and \(E_*^t\) under biologically feasible parametric restrictions. When \(p\), \(e_1\), and \(c_1\) are sufficiently small, both \(E_*^a\) and \(E_*^t\) are LAS if
\[
\frac{r\left(1+\gamma_MM^*\right)\left(K-2X^*\right)}{K\left(1 + \gamma_1 (T - T_X)^2\right)} >  d_X + \frac{\alpha\beta Y^*}{\left(X^*+\beta\right)^2}, \quad \text{and} \quad X^* < K/2;
\]
otherwise, they are unstable. 
\end{remark}
\begin{proof}
\label{proof: sufficient and necessary conditions for E_*}
The Jacobian matrix \(D\Phi (E_*^a)\) evaluated at \(E_*^a\) is given by 
\begin{equation} \label{mat:Jacobian}
\tiny
\begin{pmatrix}
       -c_1(M_{\text{in}} + e_1 p D) - M_{\text{out}} & 0 & 0 & c_1 e_1 p \max\{\Bar{M}-M,0\} \\
       \frac{r \gamma_M X \left(1 - \frac{X}{K}\right)}{1 + \gamma_1 (T - T_X)^2} & \frac{r \left(1 + \gamma_M M\right) \left(1 - \frac{2X}{K}\right)}{1 + \gamma_1 (T - T_X)^2} - \frac{\alpha \beta Y}{(\beta + X)^2} - d_X & -\frac{\alpha X}{\beta + X} & 0 \\
       -d_M Y & \frac{e \alpha Y}{(1 + \gamma_2 (T - T_Y)^2)(\beta + X)^2} & \frac{e \alpha X}{(1 + \gamma_2 (T - T_Y)^2)(\beta + X)} - (d_Y + d_M M) & 0 \\
       d_M Y & d_X & d_Y + d_M M & -p
    \end{pmatrix}
\end{equation}
and can be expressed as
\[
\begin{pmatrix}
    a_{11}^a & 0 & 0 & a_{14}^a \\
    a_{21} & a_{22} & a_{23} & 0 \\
    a_{31} & a_{32} & 0 & 0 \\
    a_{41} & a_{42} & a_{43} & a_{44}
\end{pmatrix}.
\]
The eigenvalues of \(D\Phi (E_*^a)\) are determined by solving
\[
|D\Phi (E_*^a) - \lambda I| = 0.
\]
Expanding this determinant, we obtain
\[
\begin{aligned}
    &(a_{11}^a - \lambda)(a_{44} - \lambda)(\lambda^2 - a_{22}\lambda - a_{32}a_{23}) \\
    &- a_{14}^a \Big(a_{41} \lambda^2 + (-a_{22} a_{41} + a_{21} a_{42} + a_{31} a_{43}) \lambda  \\
    &\quad - a_{23} a_{32} a_{41} + a_{23} a_{31} a_{42} - a_{22} a_{31} a_{43} + a_{21} a_{32} a_{43} \Big) = 0.
\end{aligned}
\]
Since \(a_{14}^a = c_1pe_1(\Bar{M} - M^*)\), for \(p, e_1 \ll 1\), we have \(a_{14}^a \ll 1\). For the chosen parameter values in this study, \(a_{14}^a = O(10^{-7})\). Thus, we approximate the eigenvalues by solving
\[
(a_{11}^a - \lambda)(a_{44} - \lambda)(\lambda^2 - a_{22}\lambda - a_{32}a_{23}) = 0.
\]
Furthermore, we have
\[
a_{11}^a = -c_1(M_{\text{in}} + e_1 p D^*) - M_{\text{out}} < 0, \quad a_{44} = -p < 0.
\]
Hence, two eigenvalues can be approximated as \(\lambda_1 \approx a_{11}^a\) and \(\lambda_2 \approx a_{44}\). The remaining two eigenvalues, \(\lambda_3\) and \(\lambda_4\), are determined by solving
\[
\lambda^2 - a_{22}\lambda - a_{32}a_{23} = 0.
\]
Here
\[
a_{23} = -\frac{\alpha X^*}{\beta + X^*} < 0, \quad a_{32} = \frac{e\alpha\beta Y^*}{\left(1+\gamma_2(T-T_Y)^2\right)(X^*+\beta)^2} > 0.
\]
Therefore, \(E_*^a\) is locally asymptotically stable if \(a_{22} > 0\), which gives the condition
\[
\frac{r\left(1+\gamma_MM^*\right)\left(K-2X^*\right)}{K\left(1 + \gamma_1 (T - T_X)^2\right)} >  d_X + \frac{\alpha\beta Y^*}{\left(X^*+\beta\right)^2}.
\]

Using similar arguments for \(E_*^t\), the same stability results hold.
\end{proof}


\section{Multiple timescales analysis}
\subsection{Nondimensionalization}
 \label{sec: Nondimensionalization}
 We non-dimensionalize the system (1) using the re-scaled variables as
 \begin{equation*}
    m=\gamma_M M,\,u = \frac{X}{K},\,v=\frac{Y}{K},\,w=\frac{p}{M_{\mathrm{in}}}D,\,t=rt,
\end{equation*} and obtain the non-dimensionalized system as
\begin{equation}\label{Eq:nondimension_model}
\begin{aligned}
     \frac{\mathrm{d} m}{\mathrm{d} t} &= \zeta\max\{ \left(\theta-m\right),0\}(1 + e_1w) - \sigma m,\\
     \frac{\mathrm{d} u}{\mathrm{d} t} &= \rho_1 u\left( 1 + m \right)\left(1-u\right) - \frac{ \tau uv}{\kappa + u}  - \mu_1 u,\\
     \frac{\mathrm{d} v}{\mathrm{d} t} &= \rho_2  \frac{\tau uv}{\kappa + u}  - \left(\mu_2 + \eta m\right) v,\\
    \frac{\mathrm{d}w}{\mathrm{d} t} & =  \varepsilon_2\left( \mu_1 u + \left(\mu_2 + \eta m\right) v - \delta w\right),             
\end{aligned}
\end{equation}
where the dimensionless parameters are 
\begin{equation} \label{eq:nondimen_parameter}
   \begin{aligned}
         \zeta =\frac{M_{\mathrm{in}}c_1}{r},\,&\theta =\gamma_M \bar{M},\, \sigma= \frac{M_{\mathrm{out}}}{r},\, \rho_1=\frac{1}{1+ \gamma_1 (T-T_X)^2},\,\kappa=\frac{\beta}{K},\,
  \mu_1 = \frac{d_X}{r},\, \mu_2=\frac{d_Y}{r},\\
  &\rho_2=\frac{e}{1+ \gamma_2 (T-T_Y)^2},\,\
  \tau=\frac{\alpha}{r},\,
  \eta = \frac{d_M}{r\gamma_M},\,
 \delta = \frac{M_{\mathrm{in}}}{rK},\,\varepsilon_2 =\frac{pK}{M_{\mathrm{in}}}.     
   \end{aligned}
\end{equation}
Table \ref{tab:Model_para_table} shows that detritus dynamics operate on a smaller scale than phytoplankton and zooplankton, which are smaller still than methane dynamics. To capture this hierarchy, we rescale the parameters with a small dimensionless parameter \cite{Pranali24} $0<\varepsilon_1<1$ such that $$\rho_1 =\varepsilon_1\Tilde{\rho_1},\,\tau = \varepsilon_1\Tilde{\tau},\,\mu_1=\varepsilon_1\Tilde{\mu_1},\,\mu_2=\varepsilon_1\Tilde{\mu_2},\,\eta=\varepsilon_1\Tilde{\eta}$$ and by removing the tilde we obtain a multiple timescale system as follows
\begin{equation}\label{Eq:three_timescale_model}
\begin{aligned}
     \frac{\mathrm{d} m}{\mathrm{d} t} &= \zeta\max\{ \left(\theta-m\right),0\}(1 + e_1w) - \sigma m \equiv
      F_1(m,w),\\
     \frac{\mathrm{d} u}{\mathrm{d} t} &= \varepsilon_1\left(\rho_1 u\left( 1 + m \right)\left(1-u\right) - \frac{ \tau uv}{\kappa + u}  - \mu_1 u\right) \equiv
      \varepsilon_1 F_2(m,u,v),\\
     \frac{\mathrm{d} v}{\mathrm{d} t} &= \varepsilon_1\left(\rho_2  \frac{\tau uv}{\kappa + u}  - \left(\mu_2 + \eta m\right) v\right) \equiv
      \varepsilon_1 F_3(m,u,v),\\
    \frac{\mathrm{d}w}{\mathrm{d} t} & =  \varepsilon_2\left( \mu_1 u + \left(\mu_2 + \eta m\right) v - \delta w\right)\equiv
      \varepsilon_2 F_4(m,u,v,w).         
\end{aligned}
\end{equation}
The values of the dimensionless parameters of the system \eqref{Eq:three_timescale_model} are based on the biologically feasible range of parameter values provided in Table \eqref{tab:Model_para_table}. For numerical simulations in this section, we fix the dimensionless parameters at
\begin{equation}\label{par:parameter_values_nondimension}
    \begin{aligned}
        \theta = 22.7,\,&\sigma=0.8,\,\rho_1=0.4,\,
  \kappa=0.55,\,\mu_1=0.1,\,
  \mu_2=0.08,\,
  \rho_2=0.55,\\
  &\tau=6,\,\eta=0.01,\,
 \delta=40,\,\varepsilon_1=0.1,\,
  \varepsilon_2=0.0001  
    \end{aligned}
\end{equation}
and $\zeta$ is varied between $0-8$. We now have a system with three levels of dynamics: fast, intermediate, and slow. The dynamics of methane occur on a fast timescale, while those of phytoplankton and zooplankton occur on an intermediate timescale, and the dynamics of detritus are slow. The fast dynamics evolve with respect to time \(t\), the intermediate dynamics with respect to the timescale \(t_1\), where \(t_1 = \varepsilon_1 t\), and the slow dynamics with respect to the timescale \(t_2\), where \(t_2 = \varepsilon_2 t\), such that \(0 < \varepsilon_2 \ll \varepsilon_1 < 1\). 
Therefore, with these transformations, we can rewrite the system \eqref{Eq:three_timescale_model} in terms of the intermediate and slow systems, respectively, as follows
\begin{equation}\label{Eq:intermediate_system}
\begin{aligned}
   \varepsilon_1  \frac{\mathrm{d} m}{\mathrm{d} t_1} &= \zeta\max\{ \left(\theta-m\right),0\}(1 + e_1w) - \sigma m ,\\
     \frac{\mathrm{d} u}{\mathrm{d} t_1} &= \left(\rho_1 u\left( 1 + m \right)\left(1-u\right) - \frac{ \tau uv}{\kappa + u}  - \mu_1 u\right),\\
     \frac{\mathrm{d} v}{\mathrm{d} t_1} &= \left(\rho_2  \frac{\tau uv}{\kappa + u}  - \left(\mu_2 + \eta m\right) v\right) ,\\
  \varepsilon_1  \frac{\mathrm{d}w}{\mathrm{d} t_1} & =  \varepsilon_2\left( \mu_1 u + \left(\mu_2 + \eta m\right) v - \delta w\right),           
\end{aligned}
\end{equation}
and 

\begin{equation}\label{Eq:slow_model}
\begin{aligned}
   \varepsilon_2  \frac{\mathrm{d} m}{\mathrm{d} t_2} &= \zeta\max\{ \left(\theta-m\right),0\}(1 + e_1w) - \sigma m ,\\
     \varepsilon_2\frac{\mathrm{d} u}{\mathrm{d} t_2} &= \varepsilon_1\left(\rho_1 u\left( 1 + m \right)\left(1-u\right) - \frac{ \tau uv}{\kappa + u}  - \mu_1 u\right),\\
     \varepsilon_2\frac{\mathrm{d} v}{\mathrm{d} t_2} &= \varepsilon_1\left(\rho_2  \frac{\tau uv}{\kappa + u}  - \left(\mu_2 + \eta m\right) v\right) ,\\
   \frac{\mathrm{d}w}{\mathrm{d} t_2} & =  \left( \mu_1 u + \left(\mu_2 + \eta m\right) v - \delta w\right).         
\end{aligned}
\end{equation}
The dynamics of systems \eqref{Eq:three_timescale_model}, \eqref{Eq:intermediate_system}, and \eqref{Eq:slow_model} are equivalent for \(\varepsilon_1, \varepsilon_2 \neq 0\). However, as \(\varepsilon_1\), \(\varepsilon_2\), or both approach zero, the full system reduces to specific subsystems. Thus, analyzing these subsystems is crucial for understanding the full system's dynamics.

We also give the non-dimensionalization form of the system (1) with (3) as follows:

\begin{equation}\label{Eq:terres_nondimension_model}
\begin{aligned}
     \frac{\mathrm{d} m}{\mathrm{d} t} &= \zeta (1 + e_1w) - \sigma m,\\
     \frac{\mathrm{d} u}{\mathrm{d} t} &= \rho_1 u\left( 1 + m \right)\left(1-u\right) - \tau \frac{ K }{\kappa + u}uv  - \mu_1 u,\\
     \frac{\mathrm{d} v}{\mathrm{d} t} &= \rho_2 \tau \frac{K}{\kappa + u} uv - \left(\mu_2 + \eta m\right) v,\\
    \frac{\mathrm{d}w}{\mathrm{d} t} & =  \varepsilon_2\left( \mu_1 u + \left(\mu_2 + \eta m\right) v - \delta w\right),             
\end{aligned}
\end{equation}
with $\zeta =\frac{M_{\mathrm{in}}\gamma_M}{r}$ and other dimensionless parameters are same as in \eqref{eq:nondimen_parameter}.

\section{Figures} \label{app:extra_figure}

Here we show a one-parameter bifurcation diagram of the system (1)-\eqref{Eq:aquatic_eq} over varying $T$ and $M_{\mathrm{in}}.$ The result shows that with increasing temperature the system encounters multiple transitions in the dynamics. For instance, when $M_{\mathrm{in}}$ is low, the system (1) with \eqref{Eq:aquatic_eq} exhibits two equilibrium points: the unstable extinction state \(E_0(1.08, 0, 0, 0)\) and the stable zooplankton-free equilibrium \(E_1(1.08, 0.48, 0, 0.19)\) (cf. \ref{Appendix:Stability}), for temperatures below the optimal range for plankton growth.  As temperature increases, the system undergoes a transcritical bifurcation at \(T_{\mathrm{TC1}} = 17\), giving rise to a stable coexistence equilibrium \(E_*(1.08, 0.14, 0.01, 0.06)\). A further temperature increase triggers a Hopf bifurcation at \(T_{\mathrm{H1}} = 18\), where \(E_*\) loses stability and oscillatory coexistence emerges. Further, the amplitude of oscillations of daphnia grows until its optimal temperature, beyond which it gradually decreases and encounters a second Hopf bifurcation at \(T_{\mathrm{H2}} = 22.7\), where the system stabilizes at \(E_*\). Finally, at \(T_{\mathrm{TC2}} = 23\), the system undergoes another transcritical bifurcation, reverting to \(E_1\) as the only stable equilibrium. To summarize, at low methane concentrations, plankton species coexist periodically near the optimal temperature for daphnia (\(T_Y\)). Daphnia populations peak at \(T_Y\) but decline as temperatures rise, leading to increased algae density due to reduced grazing pressure. 

   At higher methane input, algae dominate at lower temperatures (\(T < T_{\mathrm{TC1}} = 19\)), with significantly higher densities compared to the low methane case. In the narrow range \(T_{\mathrm{TC1}} < T < T_{\mathrm{H1}} = 19.4\), stable coexistence is observed. Beyond \(T_{\mathrm{H1}}\), oscillatory coexistence appears but within a narrower temperature range compared to the low methane scenario. For \(T_{\mathrm{H2}} < T < T_{\mathrm{TC2}} = 20.85\), daphnia density declines sharply, and the system settles into a zooplankton-free state. This highlights the increased sensitivity of daphnia to rising temperatures under moderate methane levels (Fig.~\ref{fig:Temperature_bifurcation}(b)). Thus our results imply that higher methane concentrations narrow the temperature range for coexistence and fasten the daphnia extinction. This indicates that prolonged exposure to moderate methane levels increases daphnia sensitivity to rising temperatures.
\begin{figure}[ht!]
    \centering
    \subfigure[Low concentration of methane]{\includegraphics[width=0.3\textwidth]{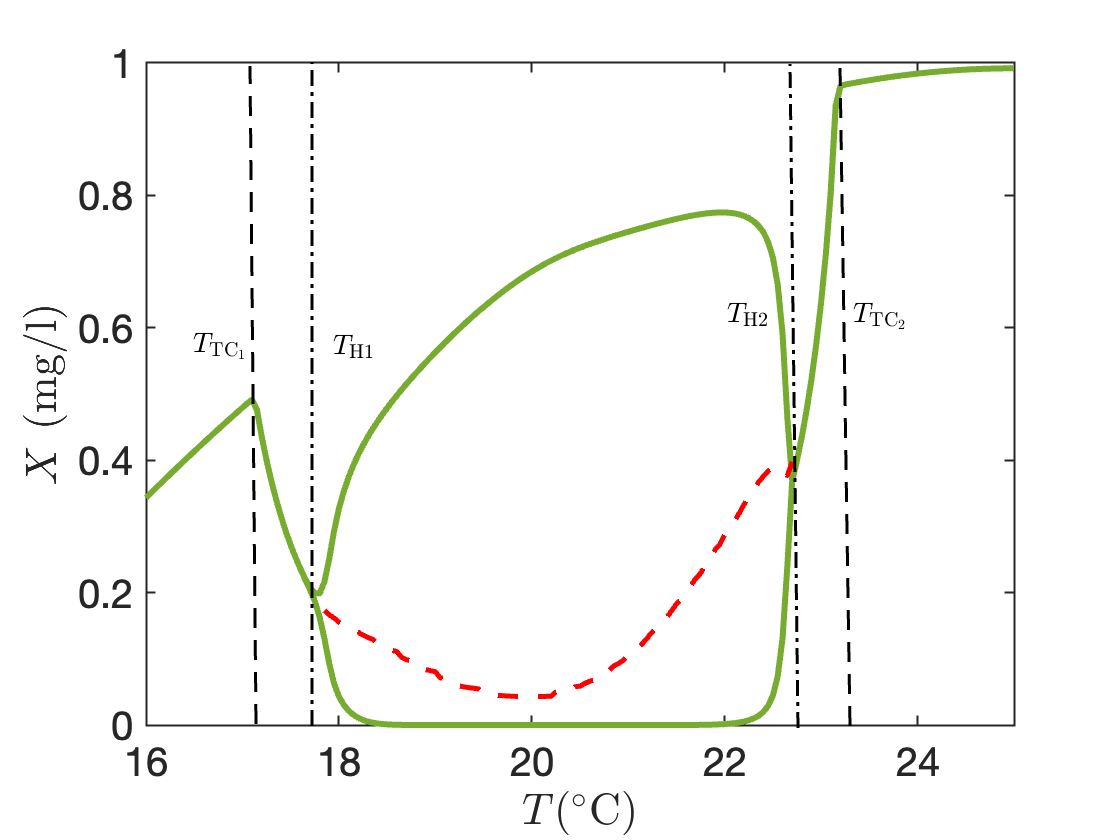}
     \includegraphics[width=0.3\textwidth]{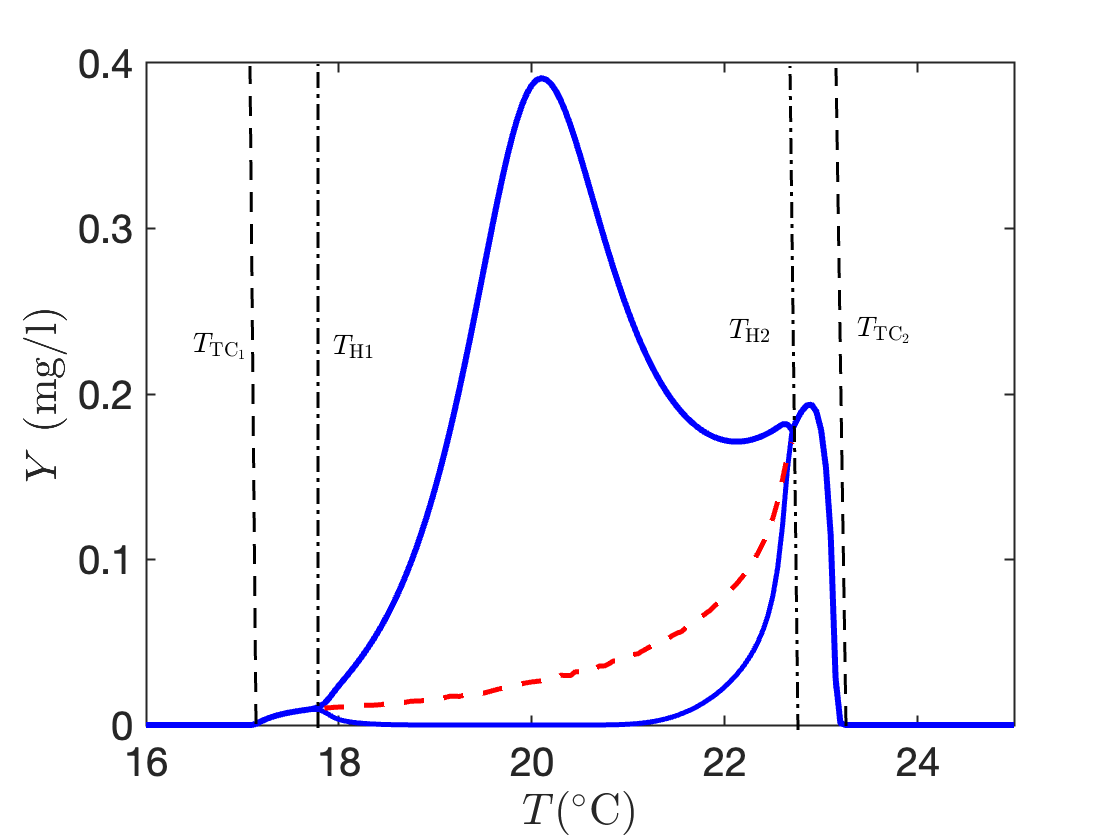}
     \includegraphics[width=0.3\textwidth]{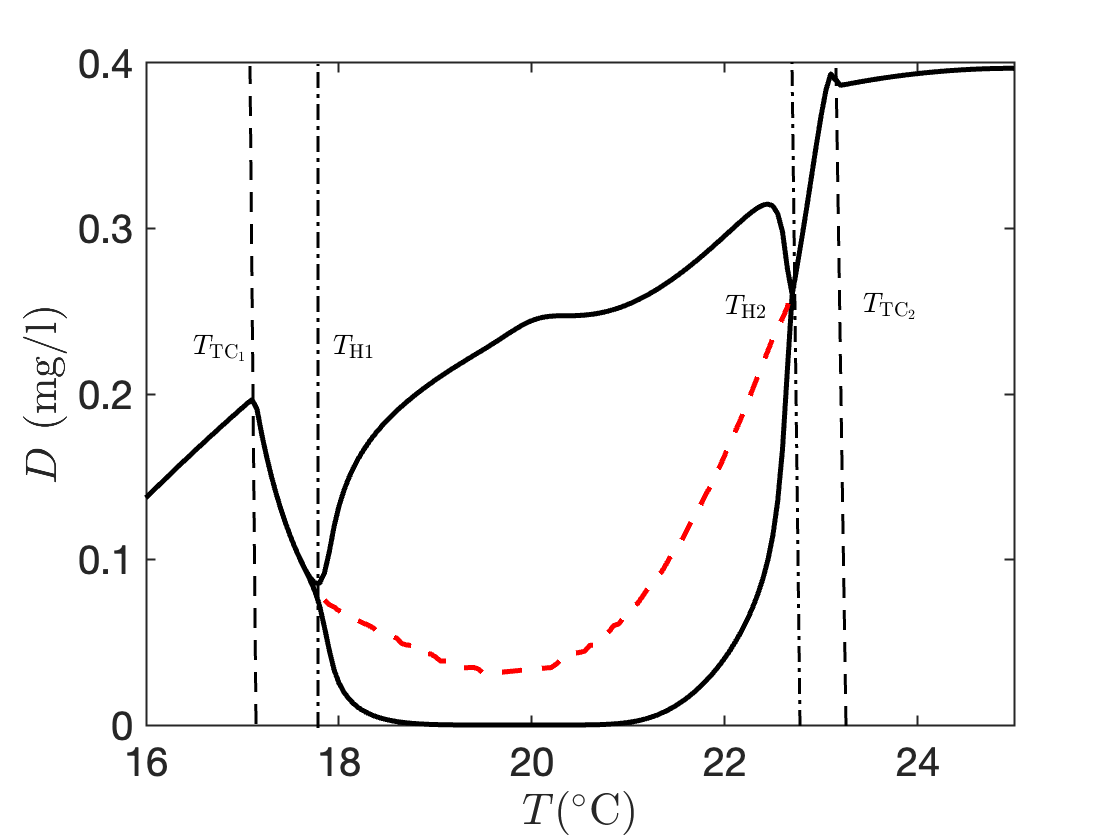}
     }

      \subfigure[Moderate concentration of methane]{\includegraphics[width=0.3\textwidth]{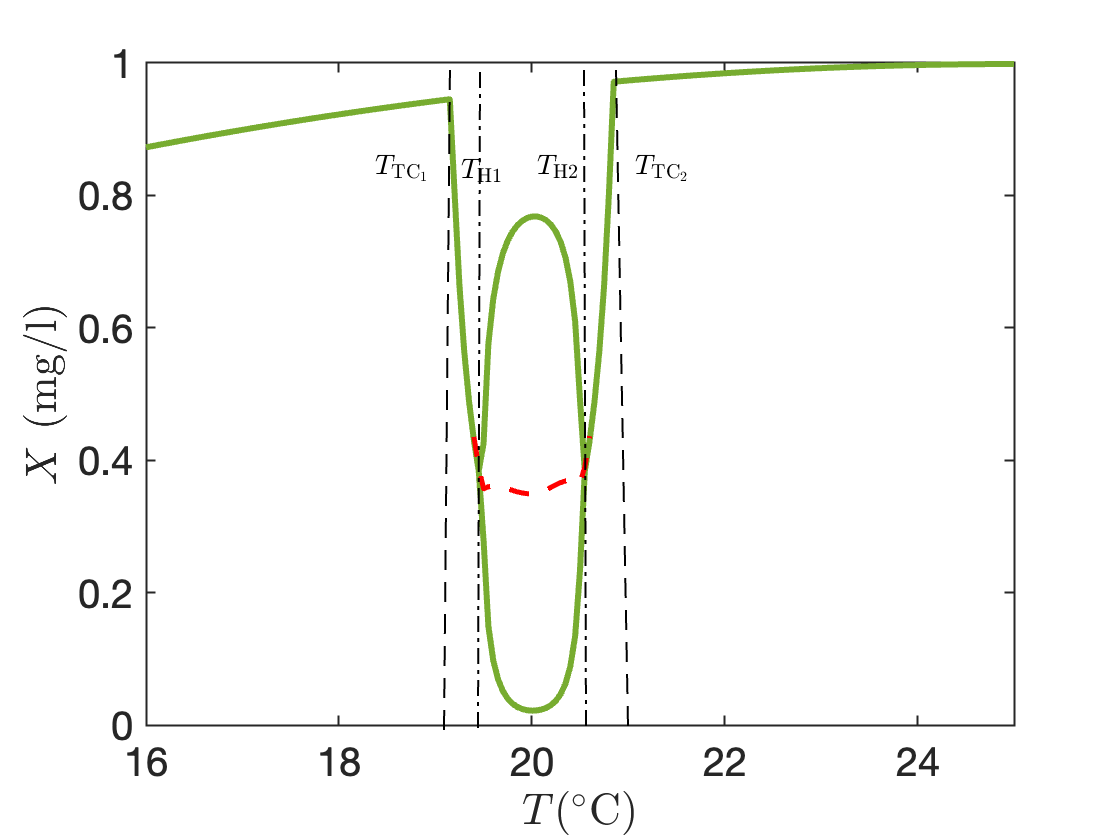}
     \includegraphics[width=0.3\textwidth]{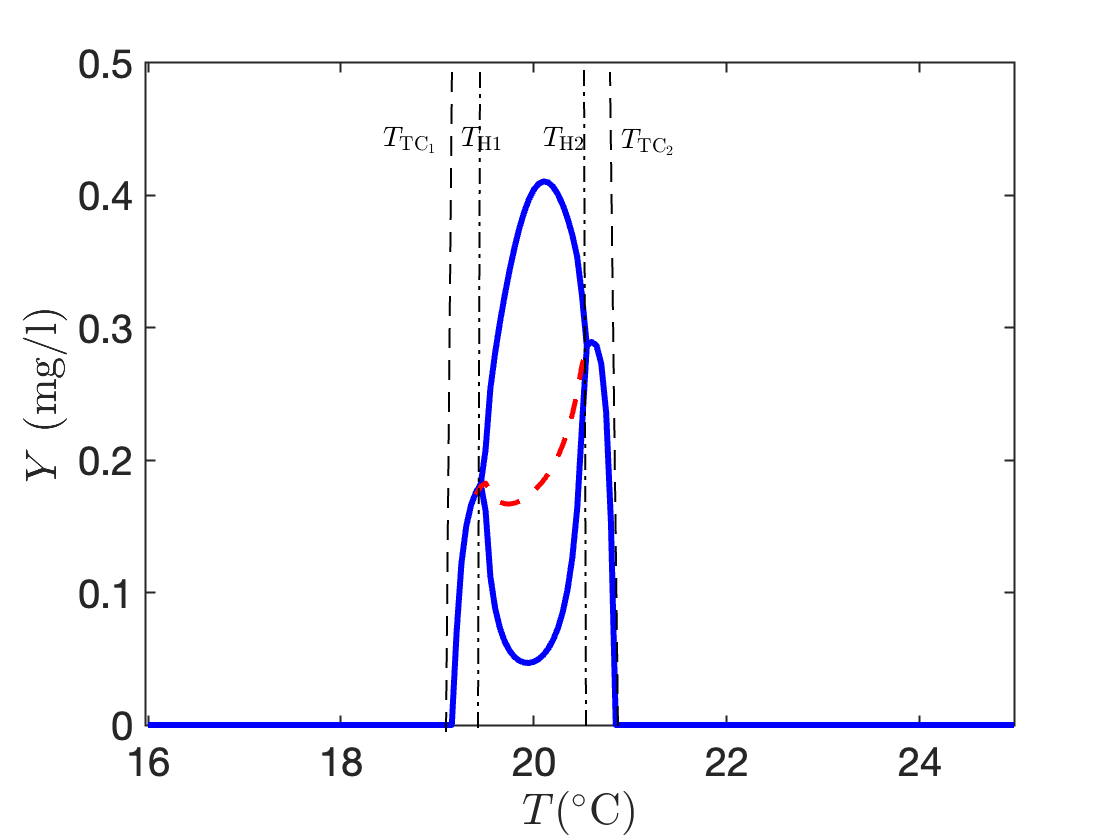}
     \includegraphics[width=0.3\textwidth]{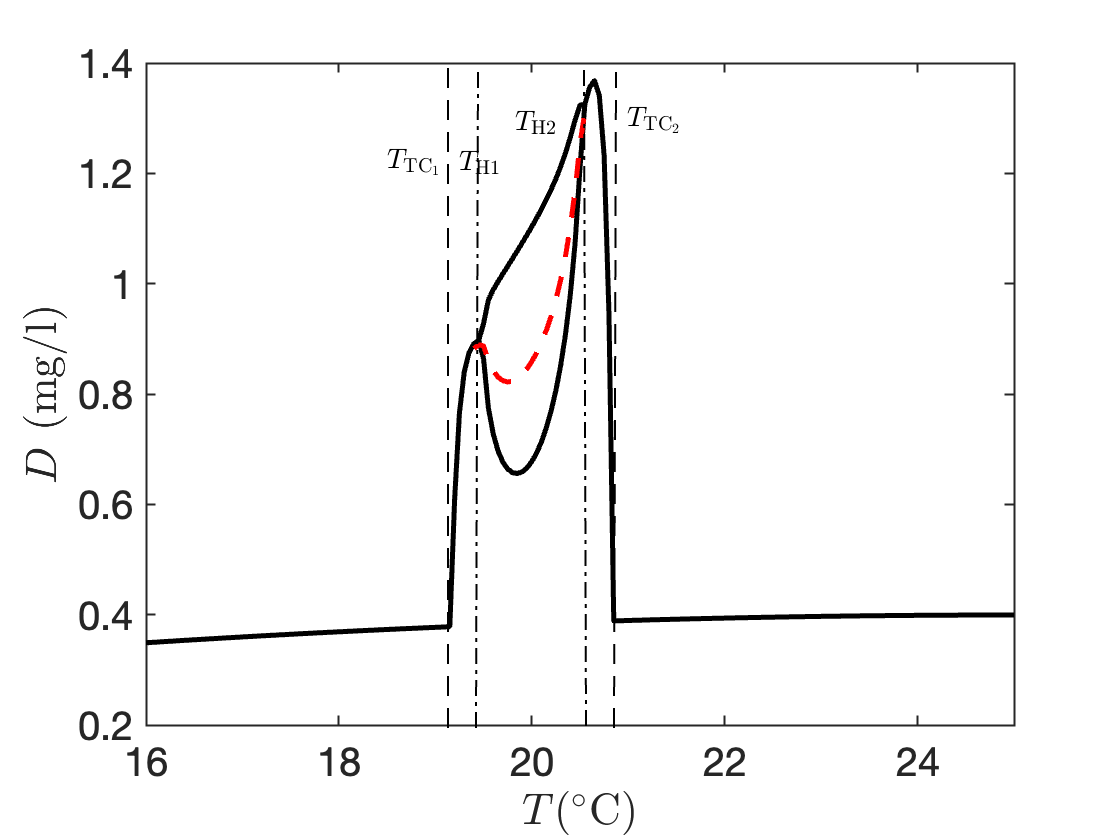}
     }   
    \caption{Bifurcation diagram of algae (green), daphnia (blue), and detritus (black) over changing temperature $T;$ (a) The external input of methane $M_{\mathrm{in}}=1,$ (b) The external input of methane $M_{\mathrm{in}}=15$. The other parameter values are fixed at \eqref{par:parameter_temp}. }
    \label{fig:Temperature_bifurcation}
\end{figure}

\end{appendix}

\end{document}